%% file: Escape-Continuous.tex
\documentclass[a4paper,UKenglish,cleveref, autoref, thm-resteat]{lipics-v2021}

\usepackage{amsmath}
\usepackage{amssymb}
\usepackage{amsthm}
\usepackage{amscd}

\usepackage{fouriernc}
\usepackage{hyperref}

\usepackage{tikz-cd}
\usetikzlibrary{decorations.pathmorphing}
\usepackage{yfonts}

\usepackage{todonotes}

\usepackage{algorithm}
\usepackage{algpseudocode}

\usepackage{letltxmacro}

\usepackage{eqparbox}
\renewcommand{\algorithmiccomment}[1]{\hfill\eqparbox{COMMENT}{\textcolor{gray}{#1}}}

\include{commands}

\title{Deciding Robust Instances of an Escape Problem for Dynamical Systems in Euclidean Space} 

\titlerunning{Deciding Robust Instances of an Escape Problem for Dynamical Systems} 

\author{Eike Neumann}{Swansea University, United Kingdom}{}{}{}

\authorrunning{E. Neumann} 

\Copyright{Eike Neumann} 

\ccsdesc[500]{Theory of computation} 

\keywords{Dynamical Systems, Computability in Analysis, Non-Linear Functions} 

\category{} 

\relatedversion{} 

\nolinenumbers
\hideLIPIcs

\EventEditors{}
\EventNoEds{1}
\EventLongTitle{}
\EventShortTitle{}
\EventAcronym{}
\EventYear{}
\EventDate{}
\EventLocation{}
\EventLogo{}
\SeriesVolume{}
\ArticleNo{}

\begin{document}
\algnewcommand\algorithmicof{\textbf{of}}
\algnewcommand\algorithmiccase{\textbf{case}}
\algdef{SE}[SWITCH]{Switch}{EndSwitch}[1]{\algorithmiccase\ #1\ \algorithmicof}{\algorithmicend\ \algorithmicswitch}%
\algdef{SE}[CASE]{Case}{EndCase}[1]{#1}{\algorithmicend\ \algorithmiccase}%
\algtext*{EndSwitch}%
\algtext*{EndCase}%

\maketitle

\begin{abstract}
    We study the problem of deciding whether a point escapes a closed subset of $\R^d$ under the iteration of a continuous map $f \colon \R^d \to \R^d$ in the bit-model of real computation.
    We give a sound partial decision method for this problem which is complete in the sense that its halting set contains the halting set of all sound partial decision methods for the problem.
    Equivalently, our decision method terminates on all problem instances whose answer is robust under all sufficiently small perturbations of the function.
    We further show that the halting set of our algorithm is dense in the set of all problem instances.
    While our algorithm applies to general continuous functions,
    we demonstrate that it also yields complete decision methods for much more rigid 
    function families: 
    affine linear systems and quadratic complex polynomials.
    In the latter case, completeness is subject to the density of hyperbolicity conjecture in complex dynamics.
    This in particular yields an alternative proof of Hertling's (2004) conditional answer to a question raised by Penrose (1989) regarding the computability of the Mandelbrot set.
\end{abstract}

\section{Introduction}

A large number of problems in a diverse range of areas,
such as program verification \cite{OWSurvey15},
numerical analysis \cite{Szegö1933, PillweinJacobi},
economics \cite{Economics},
probabilistic systems \cite{MarkovChains, Quantum},
and
biology \cite{YuChen, MelczerMezzarobba}
reduce to a problem of the following type:
Given a set $A \subseteq X$, a function $f \colon X \to X$, and a point $x_0 \in A$,
decide whether $x_0$ escapes $A$ under finitely many iterations of $f$.

Unsurprisingly, the study of automated methods for the solution of such problems has received much attention \cite{Mignotte1984,Vereshchagin, OWPositivityLowOrder14,OWPositivitySimple14,OWUltimateSimple14,DBLP:conf/icalp/NeumannO021,HolonomicPeriods20,Pillwein13,PillweinSchussler15,KauersPillwein10,GerholdKauers05,BellGerhold07,kenison2024threshold}.

Existing work on automated decision methods largely focusses on linear systems.
In fact, the scope of automated decision methods for non-linear systems appears to be quite limited:
the above problem is undecidable \cite{Tiwari} already when $f \colon \R^d \to \R^d$ is a piecewise linear function and $A$ is a polyhedron in $\R^d$.

All of the aforementioned results assume that the system is specified exactly, by means of rational or algebraic numbers.
In application areas such as engineering and the natural sciences, it may be unrealistic to assume that the system under consideration is known to perfect accuracy.
Rather, one should assume that the system is known only up to some (presumably small) error with respect to a given distance function.
In this context, one is arguably less interested in deciding the problem for a single given instance, but to exhibit a neighbourhood of the given instance for which the answer to the problem remains constant -- provided that such a neighbourhood exists.
If the given instance lies on a ``decision boundary'', \textit{i.e.}, the answer to the problem is sensitive to arbitrarily small perturbations of the instance, not much appears to be gained by solving the problem for that specific instance.
In some sense, this can also be viewed as an opportunity to increase the scope of automated decision methods,
since the aforementioned undecidability results usually occur due to instances that lie on such decision boundaries.

This has led to various problems of this type being studied in various, generally non-equivalent, formalisations of ``robust decidability'' \cite{DPLRSR21, Robust, Ratschan, Pseudo, Rounding}.

The aim of this paper is to demonstrate by means of a case study that \emph{computable analysis} constitutes a very suitable framework for the study of robust decidability questions.
Computable analysis is the study of computation on data that is not specified by a complete finite description, but rather given as an infinite sequence of increasingly accurate approximations. 
For example, a real number $x \in \R$ can be encoded by an infinite sequence $(I_n)_n$ of nested intervals with rational endpoints such that $\{x\} = \bigcap_{n \in \N} I_n$.
An algorithm for computing a real number $x \in \R$ takes as input a non-negative integer $n \in \N$ and outputs an interval $I_n$ with rational endpoints. 
The algorithm is required to ensure that the sequence $(I_n)_n$ is nested and satisfies $\{x\} = \bigcap_{n \in \N} I_n$.
An algorithm for computing a real function $f \colon \R \to \R$ takes as input an infinite sequence of intervals 
$(I_m)_m$ which is guaranteed to encode a real number $x \in \R$, and returns as output an infinite sequence of intervals $(J_n)_n$.
The algorithm is required to ensure that for all input sequences $(I_m)_m$ encoding a real number $x$, 
the output sequence $(J_n)_n$ encodes the real number $f(x)$.
The infinite sequence $(I_m)_m$ is presented to the algorithm as a ``black box'': the algorithm can query the black box for any integer $m \in \N$ to receive the interval $I_m$.
Formally, this idea can be realised by using oracle Turing machines with an ``input oracle'' \cite{Ko} or, equivalently, by using Turing machines with an infinite read-only tape that contains the entire infinite input sequence \cite{WeihrauchBook}.
It is important to observe that in the definition above, an algorithm for computing a real function operates on \emph{all} real numbers, not just on computable ones. 
This should be contrasted with \emph{Markov computability} \cite{Kushner, HoyrupRojasMarkov}, where algorithms operate only on computable objects which are presented as indexes of Turing machines that compute the objects.

One may object that the above model of computation is still unrealistic: Our initial motivation -- to study systems that are not known exactly -- suggests to consider decision methods for objects that are known only to \emph{some fixed finite precision}.
While in computable analysis objects are not given by an infinitely precise finite description,
they can still be approximated to \emph{arbitrary}, rather than fixed, finite precision.
However, an algorithm that computes a function from one space to another can be automatically lifted to one that computes the range of the function over any given compact subset of the space \cite[Proposition 11]{PaulyRepresented}, automatically yielding an effective method for operating on objects that are known only to some fixed precision. 
See also \cite[Section 2]{DPLRSR21} for a discussion of this in the context of decision problems.

We study the following very general instantiation of the decision problem mentioned above:
Given a continuous function $f \colon \R^d \to \R^d$, a closed set $A \subseteq \R^d$, and a point $x_0 \in A$, determine whether the point escapes $A$ under iteration of $f$. To make this precise, we introduce standard encodings of points, sets, and functions, which we discuss in detail in Section \ref{Section: Encodings}.
Here, the point escapes if and only if there exists $n \in \N$ such that $f^n(x_0) \notin A$ -- \textit{i.e.}, the point is said to escape if it leaves the set after some number of function applications, even if it re-enters the set after further applications of $f$.
We will refer to this problem as the \emph{Point Escape Problem}.

Since the inputs to an algorithm are given as infinite sequences of increasingly accurate approximations, there is no hope to obtain a correct algorithm that halts on every input for \emph{any} non-trivial problem of the above type: every computable function must be continuous, and every continuous function from a connected space to a two-point space is constant.
Since no \emph{total} algorithm exists for deciding the problem in question, the next best al\-ter\-na\-tive is to ask for a partial algorithm that halts on as many problem instances as possible.
Accordingly, we will call a partial algorithm for deciding a problem \emph{complete} if its halting set contains the halting set of all correct partial algorithms deciding the same problem.
This is equivalent to asking that the algorithm halt on all problem instances whose answer is stable under small perturbations. 
See Section \ref{Section: Decision Methods} for a more formal discussion and \cite[Proposition 2.1]{DPLRSR21} for a proof of this equivalence.
Observe that for discrete spaces such as the natural numbers, an algorithm is complete in this sense if and only if it halts on all inputs, so that the above definition of completeness is a generalisation of the usual definition in classical computability.

Our algorithm for solving the Point Escape Problem is rather straightforward: We keep track of an overapproximation $O$ of a finite initial segment 
$\left\{x_0, f(x_0), \dots, f^{N}(x_0)\right\}$ of the orbit, as well as of overapproximations $Q_i$ of the individual points $\left\{f^i(x_0)\right\}$ in the orbit.
If we witness that $Q_i$ is separated from the closed set $A$, then we conclude that $x_0$ escapes under $f$.
If we witness that $O$ is included in the interior of $A$ and we can find an overapproximation $O'$ of the image $f(O)$ with 
$O' \subseteq O$, then we conclude that the point $x_0$ must be trapped.
If the overapproximations we have computed do not allow us to draw either conclusion, we compute more accurate overapproximations to longer initial segments of the orbit, until we are able to make a decision.

Our main contribution is to show that this simple algorithm is complete in the above sense.
The key idea is to show that the search for an invariant set for $f$ is guaranteed to terminate provided that $f$ has a \emph{robust invariant} whose interior contains $f(x_0)$. A compact set $V \subseteq \R^d$ is called a robust invariant for $f$ if $f(V) \subseteq \intr{V}$, where $\intr{V}$ denotes the interior of $V$.
We show that if $x_0$ is trapped but $f(x_0)$ is not contained in a robust invariant, then $x_0$ escapes under arbitrarily small perturbations of $f$ and $A$.

We will further show that if $x_0$ is trapped but $f(x_0)$ is not contained in a robust invariant and $A \neq \R^d$, then 
$x_0$ escapes $A$ under arbitrarily small perturbations of $f$ alone, with $A$ being fixed.
This implies that our algorithm, which takes $x_0$ and $A$ as inputs, 
is complete for all fixed $x_0$ and $A \subsetneq \R^d$, when only the function $f$ is given as an input.
Thus, a bespoke algorithm for fixed special sets or initial points, say $x_0 = 0$ or $A = [0,1]^d$ will not halt on more functions than our general algorithm.

Since a complete algorithm is only required to halt on robust instances, one is lead to the problem of determining the ``size'' of the robust instances. In the worst case, a problem may not have any robust instances at all, so that a complete decision method is given by the algorithm that never halts. 
We show that the halting set of our algorithm is ``large'', in the sense that the algorithm halts on a dense set of inputs and that the set of trapped problem instances is the closure of its interior.

Finally, which problem instances are robust depends on the class of functions and how these functions are represented (by our aforementioned result it surprisingly does not essentially depend on how the sets or initial points are represented). 
We allow arbitrary continuous functions as inputs, and we will represent these using the weakest representation that makes function evaluation uniformly computable. 

In general, if the problem is restricted to a smaller class of functions, potentially with a representation that induces a stronger topology, then previously non-robust instances may become robust, so that our algorithm will in general fail to be complete for restrictions of the Point Escape Problem to smaller function classes.

Indeed, The topology induced by our representation of continuous functions is the topology of uniform convergence on compact sets. 
In this topology, every neighbourhood of a function $f$ contains functions $g$ that agree with $f$ on some compact set, but differ arbitrarily from $f$ outside this compact set.
For this reason, our algorithm will fail to halt on very simple-looking problem instances where the given set is unbounded.
Consider for example the problem instance $(A,f,x_0)$ consisting of the set $A = \Set{x \in \R}{x \geq 0}$, the function $f(x) = 2x$, and the initial point $x_0 = 1$.
Since $f$ does not map any compact set into its interior, our algorithm will fail to halt on this problem instance.

However, we will show that systems such as the above can be effectively treated by our algorithm by first applying a suitable compactification.
More specifically, we show that the problem of deciding whether a point escapes a polyhedron (which is in general unbounded) under an affine linear map reduces to the Point Escape Problem for continuous functions in such a way that robust problem instances get mapped to robust problem instances. We conjecture that this reduction extends to more general non-linear systems that are sufficiently well behaved ``at infinity''.

To discuss the applicability of our algorithm to a class of much more rigid non-linear systems, we consider the problem of deciding whether the origin in the complex plane has an unbounded orbit under a quadratic polynomial, represented by a single complex parameter. The existence of a complete algorithm for this problem is closely related to the computability of the Mandelbrot set \cite{hertling9}. We give a reduction to the Point Escape Problem that maps robust instances to robust instances if and only if the hyperbolicity conjecture holds true.

The rest of the paper is structured as follows: in Section \ref{Section: Preliminaries} we review the relevant background material from computable analysis,
discuss completeness of decision methods in the context of computable analysis, and formally introduce our representations of points, sets, and functions.
Section \ref{Section: Algorithm and Main Results} contains a formal description of our main algorithm and our main results: that the algorithm is correct, complete, and generically terminating.
In Sections \ref{Section: Linear Systems} and \ref{Section: Mandelbrot} we describe our applications to linear and quadratic systems respectively.

\section{Preliminaries}\label{Section: Preliminaries}

\subsection{Computable Analysis}

We review some basic definitions and results from computable analysis. 
Good introductions to the subject are given in \cite{WeihrauchBook,BrattkaHertling21,Ko,BrattkaPresser,PaulyRepresented}.
The key idea for computing with first-order objects such as real numbers and real functions is to encode them via integer sequences.
Let $\NN$ denote the space of all sequences $(p_n)_n$ of non-negative integers.
We make $\NN$ into a topological space by endowing it with the product topology, which is generated by the distance function 
$d\left((p_n)_n, (q_n)_n\right) = 2^{-\inf\left\{n \mid p_n \neq q_n\right\}}$.
A sequence $(p_n)_n$ is called \emph{computable} if there exists a Turing machine which takes as input a natural number $n$ and returns as output the natural number $p_n$.

An oracle Turing machine \emph{computes} a partial function 
$f \colon \subseteq \NN \to \NN$
if for all oracles $\Phi$ with $(\Phi(n))_n \in \dom f$
and all natural numbers $N$,
the machine halts in finite time on input $N$ with oracle $\Phi$ and outputs the number
$q_N$, where $(q_m)_m = f((\Phi(n))_n)$.
Observe that the oracle $\Phi$ is not a fixed function, but part of the input.
Further, observe that we do not constrain the machine's behaviour for oracle-inputs outside the domain of $f$.
A partial function $f \colon \subseteq \NN \to \NN$ is called \emph{computable} if it is computed by some Turing machine.

A \emph{represented space} $(X, \delta_X)$ is a set $X$ together with a partial surjective map $\delta_X \colon \subseteq \NN \to X$.
We will usually write just $X$ for the represented space $(X, \delta_X)$ when $\delta_X$ is implicit or inferable from the context.
Let $X$ be a represented space, and let $x \in X$.
A point $p \in \NN$ is called a \emph{name} of $x$ if $p \in \dom \delta_X$ and $\delta_X(p) = x$.
The point $x$ is called \emph{computable} if it has a computable name.

Let $f \colon X \to Y$ be a function between represented spaces.
A partial function 
$F \colon \subseteq \NN \to \NN$
is called a \emph{realiser} of $f$
if $\dom F \supseteq \dom \delta_X$
and $\delta_Y \circ F = f \circ \delta_X$.
The function $f$ is called \emph{computable} if it has a computable realiser.
It is called \emph{continuous} if it has a continuous realiser (with respect to the relative topology induced by the product topology on $\NN$).

If $X$ is a set and $\delta_X \colon \subseteq \NN \to X$ and $\delta'_X \colon \subseteq \NN \to X$ are representations,
we say that $\delta_X$ and $\delta_X'$ are \emph{computably equivalent} or just \emph{equivalent} if 
the identity $\id_X \colon X \to X$ on $X$ is computable as a map $\id_X \colon (X,\delta_X) \to (X,\delta'_X)$
and as a map $\id_X \colon (X,\delta_X') \to (X,\delta_X)$ between represented spaces.
By replacing ``computable'' with ``continuous'' in the above definition, we obtain the definition of \emph{topological equivalence}.

Any representation $\delta_X \colon \subseteq \NN \to X$ induces a topology on $X$, namely the final topology,
where $U \subseteq X$ is open if and only if $\delta_X^{-1}(U)$ is an open subset of $\dom \delta_X$ (with respect to the relative topology induced by the product topology on $\NN$).
We will call this topology \emph{the topology of the represented space $X$}.
For a function $f \colon X \to Y$ between represented spaces, we have two a priori distinct notions of continuity available:
continuity in the sense of having a continuous realiser, and \emph{topological continuity} in the sense of being continuous 
with respect to the topologies on $X$ and $Y$.
Any continuous function is topologically continuous, but the converse is false in general.
A representation $\delta_X \colon \subseteq \NN \to X$ is called \emph{admissible} if all topologically continuous functions 
$f \colon Y \to X$, where $Y$ is a represented space, are continuous (cf.~\cite[Theorem 36]{PaulyRepresented}).
All representations we consider in this paper will be admissible.
Beyond this, we will only consider representations that are \emph{open maps}, \textit{i.e.} representations where every prefix of a name defines an open subset of the represented space.
Intuitively, this means that all names contain the same amount of information.

\subsection{Complete Decision Methods}\label{Section: Decision Methods}

An algorithm over discrete data is called \emph{complete} if it halts on all inputs.
A set is said to be \emph{decidable} if there exists a complete algorithm that correctly determines whether a given input belongs to the set.
If we apply the same definition to continuous data, we often end up with a trivial notion: if $X$ is a represented space such as $\R^d$ or $C(\R^d, \R^d)$ whose associated topology is connected, then the only decidable subsets of $X$ in  in the above sense are the empty set and $X$ itself, for a decision method defines a computable function of type $X \to \{0,1\}$ and this function must be continuous.

To obtain a more meaningful notion, the definition of ``complete algorithm'' must be extended in a different way.
Let $X$ be a represented space.
Let $A \subseteq X$.
A \emph{(partial) decision method} for $A$ is an algorithm that takes as input the name of a point $x \in X$ and either runs forever or halts in finite time. 
Upon halting, the algorithm is required to output the integer $1$ if $x \in A$ and the integer $0$ if $x \notin A$.
The \emph{halting set} of a decision method is the set of all $x \in X$ such that the algorithm halts on all names of $x$ (for open representations, a decision method that halts on \emph{some} name of $x$ automatically extends to a decision method that halts on all names of $x$).
Letting $\partial A$ denote the boundary of $A$, it is easy to see that the halting set of a decision method is contained in $X \setminus \partial A$.
A decision method for $A$ is \emph{complete} if its halting set is equal to $X \setminus \partial A$.
A subset $A \subseteq X$ which admits a complete decision method is called \emph{maximally partially decidable} or simply \emph{decidable}.
Observe that for any space $X$ which carries the discrete topology, such as $X = \N$, any subset of $X$ has empty boundary, so that our definition of decidability agrees with the usual definition.
For spaces such as $X = \R^d$, this definition of ``decidability'' yield a considerably richer structure than the trivial one induced by the na\"ive direct generalisation.
For the spaces we are interested in, completeness can be characterised as a kind of optimality:
If $X$ admits an open representation, then a decision method for $A \subseteq X$ is complete if and only if its halting set contains the halting set of every decision method for $A$. 
This definition is further discussed and motivated in \cite{DPLRSR21}.

We will call a problem instance $x \in X$ a \emph{robust instance} if $x \notin \partial A$, and a \emph{boundary instance} otherwise.
Thus, an algorithm is complete if and only if it halts on all robust instances.
Observe that if we witness an algorithm halting on a given problem instance $x \in X$, then we can identify a prefix $p \in \N^*$ of the given name of $x$ such that the algorithm halts on all names that extend $p$. 
When the representation of $X$ is open, this prefix corresponds to an open neighbourhood of $x$ where the answer to the problem is constant.

\subsection{Encoding Points, Sets, and Functions}\label{Section: Encodings}

We describe representations of points in $\R^d$, continuous functions $\R^d \to \R^d$, and closed and compact subsets of $\R^d$.
Our representations of points and functions are up to equivalence the usual standard representations of computable analysis that can be found in the literature \cite{WeihrauchBook}.
Our representations of sets are equivalent to a join of certain standard representations.

Before we introduce our representations, let us introduce some standard notation and terminology.
Throughout, we will work with the supremum norm 
\[
    \norm{x}_{\infty} = \max \Set{|x_i|}{i \in \{1, \dots, d\}}
\]
and its induced metric $d(x,y) = \norm{x - y}_\infty$
on $\R^d$.
For a point $x \in \R^d$ and positive number $r > 0$, we let 
$
    B(x, r) = \Set{y \in \R^d}{\norm{x - y}_\infty < r}
$
and 
$
    \overline{B}(x, r) = \Set{y \in \R^d}{\norm{x - y}_\infty \leq r}
$
denote the open and closed ball of radius $r$ about $x$ respectively.
In $\R^d$, the closed ball of radius $r$ about $x$ is the closure of the open ball of radius $r$ about $x$.
More generally, for a set $S \subseteq \R^d$, we let 
\[
    B(S, r) = \Set{y \in \R^d}{\exists x \in S. \norm{x - y}_\infty < r} = \bigcup_{x \in S} B(x, r)
\]
and 
\[
    \overline{B}(S, r) = \Set{y \in \R^d}{\exists x \in S. \norm{x - y}_\infty \leq r} = \bigcup_{x \in S} \overline{B}(x, r).
\]
Observe that $B(S, r)$ is always an open set, while $\overline{B}(S, r)$ is not necessarily a closed set.
In particular, $\overline{B}(S, r)$ is not necessarily equal to the closure of $B(S, r)$.
This is the case, however, if the set $S$ is closed. For a proof, see Proposition \ref{Proposition: closed balls about closed sets are closure of open balls} in the appendix.

For continuous maps $f, g \colon \R^d \to \R^d$ we let
$
    \norm{f - g}_{\infty} = \sup \Set{\norm{f(x) - g(x)}_{\infty}}{x \in \R^d} \in [0, +\infty].
$
For a compact set $K \subseteq \R^d$ we define
$
    \norm{f - g}_{\infty, K} = \sup \Set{\norm{f(x) - g(x)}_{\infty}}{x \in K} \in [0, +\infty).
$
The set $K$ is allowed to be empty in this definition.
We define a distance function on the space $C\left(\R^d, \R^d\right)$ of all continuous maps by 
\begin{equation}\label{eq: distance on functions}
    d(f, g) = \sum_{n \in \N} 2^{-n} \min\left\{1, \norm{f - g}_{\infty, \overline{B}\left(0, 2^n\right)}\right\}.
\end{equation}
We choose this metric since its induced topology is the same as the topology induced by the weakest representation of the space of continuous functions that renders evaluation computable -- see Proposition \ref{Proposition: properties of function-space representation}.

In order to formulate our main algorithm, it will be convenient to work with subdivisions of $\R^d$ into dyadic rational cubes of uniform size.
To facilitate this, we choose representations based on such subdivisions.

Let 
$
    \QD^d_n =
        \Set{
            \left[\frac{x_1}{2^n}, \frac{x_1 + 1}{2^n}\right]
            \times 
            \dots 
            \times 
            \left[\frac{x_d}{2^n}, \frac{x_d + 1}{2^n}\right]
        }
        {
            x_1,\dots,x_d \in \Z
        }
$
denote the set of all $d$-dimensional dyadic rational cubes of side-length $2^{-n}$.
The sets $\QD^d_n$ subdivide $\R^d$ into a cubical mesh of mesh-width $2^{-n}$,
with the interiors of distinct cubes being disjoint.
Let $\QD^d_* = \bigcup_{n \in \N} \QD^d_n$.

For a subset $A \subseteq \QD^d_*$ of $\QD^d_*$, we write $|A| = \bigcup_{Q \in A} Q \subseteq \R^d$. 
For a set $A \subseteq \R^d$, we write 
$
    \QD^d_n(A) = \Set{Q \in \QD^d_n}
                   {Q \cap A \neq \emptyset}.
$
For sets $A$ and $B$, we write $A \subf B$ to indicate that $A$ is a finite subset of $B$.

For a set $A \subseteq \R^d$, we let $\diam A = \sup \Set{d(x,y)}{x, y \in A}$ denote its \emph{diameter}.
We write $\intr{A}$ for the \emph{interior} of $A$, $\overline{A}$ for the \emph{closure} of $A$,
and $\partial A$ for the \emph{boundary} of $A$.
For sets $A, B \subseteq \R^d$, we write $A \Subset B$ if $\overline{A} \subseteq \intr{B}$.
It is clear that for finite subsets $A$, $B$ of $\QD^d_*$ the relations 
$|A| \subseteq |B|$ and $|A| \Subset |B|$ are decidable.

We will now introduce our representations of points, sets, and functions.
Elements of $\QD^d_*$ and finite sets of elements of $\QD^d_*$ can be coded by natural numbers.
We will fix appropriate encodings, without making them explicit.

\begin{definition}\label{Definition: rho^d}
Define a representation $\rho^d \colon \subseteq \NN \to \R^d$ as follows:

Up to coding, a $\rho^d$-name of a point $x \in \R^d$ is a sequence $(X_n)_n$ of finite sets $X_n \subf \QD^d_n$ satisfying
$x \in \intr{|X_n|}$ for all $n$,
$|X_{n + 1}| \subseteq |X_n|$ for all $n$,
and $\diam |X_n| \to 0$ as $n \to \infty$.
\end{definition}

\begin{proposition}
    The representation $\rho^d$ is open and admissible and induces the standard Euclidean topology on $\R^d$,
    \textit{i.e.}, the topology generated by the norm $\norm{\cdot}_\infty$.
\end{proposition}
\begin{proof}
    See Proposition \ref{Proposition: rho-d admissible and open (Appendix)} in Appendix \ref{Appendix: Representations}.
\end{proof}

A code of a multi-valued function $F \colon \QD^d_n\left([-2^n, 2^n]^d\right) \rightrightarrows \QD^d_n$ with finite values is an integer that encodes a list of the form 
$\left\langle \left(Q_0, \left\langle R_{0,0}, \dots, R_{0, N_0} \right\rangle\right), \dots, \left(Q_M, \left\langle R_{M,0}, \dots, R_{M, N_M} \right\rangle\right)\right\rangle$
where the sequence $(Q_i)_i$ contains every element of the finite set $\QD^d_n\left([-2^n, 2^n]^d\right)$ exactly once, and for all $i,j$ we have $R_{i,j} \in \QD^d_n$.
Such a list encodes the function that sends the cube $Q_i$ to the finite set $\left\{ R_{i,0}, \dots, R_{i, N_i} \right\}$.

\begin{definition}
Define a representation $\left[\rho^d \to \rho^d\right] \colon \subseteq \NN \to C(\R^d, \R^d)$ as follows:

Up to coding, a $\left[\rho^d \to \rho^d\right]$-name of a continuous function $f \colon \R^d \to \R^d$ is a sequence $(F_n)_n$ of set-valued maps 
$F_n \colon \QD^d_n\left([-2^n, 2^n]^d\right) \rightrightarrows \QD^d_n$
with finite values
satisfying the following requirements:
\begin{enumerate}
    \item 
    $f(Q) \subseteq \intr{|F_n(Q)|}$ for all $Q \in \QD^d_n\left(\left[-2^n, 2^n\right]^d\right)$.
    \item 
        If $Q \in \QD^d_n\left(\left[-2^n, 2^n\right]^d\right)$ and 
        $Q' \in \QD^d_{n + 1}\left(\left[-2^{n + 1}, 2^{n + 1}\right]^d\right)$ with 
        $Q' \subseteq Q$, then 
        $\left|F_{n + 1}(Q')\right| \subseteq \left|F_n(Q)\right|$.
    \item 
        If $(Q_n)_{n \geq m}$, $m \in \N$, is a sequence of cubes with
        $Q_n \in \QD^d_n\left(\left[-2^n, 2^n\right]^d\right)$,
        $Q_{n + 1} \subseteq Q_n$, 
        and 
        $\diam Q_n \to 0$ as $n \to \infty$,
        then 
        $\diam |F_n(Q_n)| \to 0$ as $n \to \infty$.
\end{enumerate}
\end{definition}

\begin{proposition}\label{Proposition: properties of function-space representation}
    The representation $\left[\rho^d \to \rho^d\right]$ is open and admissible and induces the compact-open topology on $C\left(\R^d, \R^d\right)$.
    This topology further coincides with the topology induced by the metric \eqref{eq: distance on functions}.
\end{proposition}
\begin{proof}
    See Proposition \ref{Proposition: rho-d-rho-d admissible and open (Appendix)} in Appendix \ref{Appendix: Representations}.
\end{proof}

For a set-valued function $G \colon X \rightrightarrows Y$ and a set $A \subseteq X$, we write 
$G(A) = \bigcup_{x \in A} G(x) \subseteq Y$.

The encoding of continuous functions makes function evaluation on points uniformly computable:
if $(X_n)_n$ is a name of a point $x \in \R^d$ and $(F_n)_n$ is a name of a function $f \colon \R^d \to \R^d$
then $(F_n(X_n))_n$ is a name of the point $f(x)$ -- potentially up to shifting the sequence to make $F_n$ defined on $X_n$ for all $n$.
Further, if $\delta$ is any representation of $C(\R^d, \R^d)$ that renders evaluation computable, then there exists an algorithm which translates
a $\delta$-name of $f$ to a name $(F_n)_n$ of $f$ with respect to the above representation. For a proof see \cite{WeihrauchBook, SchroederPhD}.

Continuous functions on compact sets are uniformly continuous.
It is a folklore result in computable analysis that this fact is effectively witnessed
by the above encodings (cf.~also \cite[Theorem 2.24]{Ko}):

\begin{lemma}\label{lemma: modulus of continuity witnessed by representation}
    Let $f \colon \R^d \to \R^d$ be a continuous function.
    Let $K$ be a compact set.
    Then for all $n$ there exists $m$ such that for all 
    $Q \in \QD^d_m(K)$ we have 
    $\operatorname{diam} |F_n(Q)| < 2^{-n}$.
\end{lemma}
\begin{proof}
    By definition, for all $x \in K$ there exists $Q_x \in \QD^d_{m_x}(K)$ with 
    $\operatorname{diam} |F_{m_x}(Q_x)| < 2^{-n}$.
    By monotonicity, if $\QD^d_{m} \ni Q \subseteq \intr{Q}_x$, then 
    $\operatorname{diam} |F_{m}(Q)| \leq \operatorname{diam} |F_{m_x}(Q_x)| < 2^{-n}$.

    By compactness of $K$, the cover 
    $
        \bigcup_{x \in K} \intr{Q}_x
    $
    must have a finite subcover 
    $
        \intr{Q}_{x_1}, \dots, \intr{Q}_{x_s}.
    $
    Now, there exists a number $m \in \N$ such that every $Q \in \QD^d_{m}(K)$ is contained in some $\intr{Q}_{x_i}$ with $i \in \{1, \dots, s\}$.
    The claim follows.
\end{proof}

A closed set $A$ will be represented by a sequence of lists of cubes that exhaust its complement $\R^d \setminus A$,
together with a sequence of lists of cubes that exhaust its interior $\intr{A}$.
This is equivalent to the joint of the standard representation of $A$ as a closed set with negative information (cf.~\cite[Definition 5.1.1]{WeihrauchBook})
and the standard representation of $\intr{A}$ as an open set with positive information (cf.~\cite[Definition 5.1.15]{WeihrauchBook}).

\begin{definition}
Define a representation $\alpha^d \colon \subseteq \NN \to \mathcal{A}(\R^d)$ of the set of closed subsets of $\R^d$ as follows:
Up to coding, a $\alpha^d$-name of a closed set $A \subseteq \R^d$ is a sequence of pairs 
$
    (I^A_n, E^A_n)_n
$
of finite sets $I^A_n, E^A_n \subf \QD^d_n$, satisfying the following properties:
\begin{enumerate}
    \item $|I^A_n| \Subset A$ for all $n$.
    \item $\intr{A} = \bigcup_{n \in \N} \intr{\left|I^A_n\right|}$.
    \item $|E^A_n| \Subset \R^d \setminus A$ for all $n$.
    \item $\R^d \setminus A = \bigcup_{n \in \N} \intr{\left|E^A_n\right|}$.
    \item If $Q \in I^A_n$ and $Q' \in \QD^d_{n + 1}$ with $Q' \subseteq Q$, then $Q' \in I^A_{n + 1}$.
    \item If $Q \in E^A_n$ and $Q' \in \QD^d_{n + 1}$ with $Q' \subseteq Q$, then $Q' \in E^A_{n + 1}$.
\end{enumerate}
\end{definition}

\begin{proposition}
    The representation $\alpha^d$ is open and admissible.
\end{proposition}
\begin{proof}
    See Proposition \ref{Proposition: alpha is admissible and open (Appendix)} in Appendix \ref{Appendix: Representations}.
\end{proof}

Observe that if $(X_n)_n$ is a name of a point $x \in \R^d$ and if $(I^A_n, E^A_n)_n$ is a name of a closed set $A \subseteq \R^d$, then 
$x \in \intr{A}$ if and only if for all large $n$ we have $X_n \Subset I^A_n$
and 
$x \in \R^d \setminus A$ if and only if for all large $n$ we have $X_n \Subset E^A_n$.

This property extends to arbitrary compact sets:
If $(I^A_n, E^A_n)_n$ is a name of a closed set, and $K$ is a compact set which is contained in $\intr{A}$ or $\R^d \setminus A$,
then this containment is witnessed by a finite initial segment of the name $(I^A_n, E^A_n)_n$.

\begin{proposition}\label{Proposition: witnessing containment of compact sets}
    Let $K \subseteq \R^d$ be a compact set.
    Let $(I^A_n, E^A_n)_n$ be a name of a closed set $A$.
    Then $K \subseteq \intr{A}$ if and only if for all large $n \in \N$, 
    the set
    $K$ is contained in $\intr{\left|I^A_n\right|}$.
    We have $K \subseteq \R^d \setminus A$ if and only if for all large $n \in \N$,
    the set
    $K$ is contained in $\intr{\left|E^A_n\right|}$.
\end{proposition}
\begin{proof}
    We show the claim for $K \subseteq \intr{A}$.
    The second claim is proved analogously.
    If $K \subseteq \intr{\left|I^A_n\right|}$ for some $n$ 
    then $K \subseteq \intr{A}$ by definition.
    Now, assume that $K \subseteq \intr{A}$.
    Since 
        $\intr{\left|I^A_n\right|} \subseteq \intr{\left|I^A_{n + 1}\right|}$
    for all $n$, it suffices to show that there exists $n \in \N$ such that 
    $K \subseteq \intr{\left|I^A_n\right|}$.
    We have 
    $
        K \subseteq \intr{A} = \bigcup_{n \in \N} \intr{\left|I^A_n\right|}
    $
    Since $K$ is compact, the open cover on the right-hand side must have a finite subcover.
    The claim follows.
\end{proof}

\section{The Point Escape Problem}\label{Section: Algorithm and Main Results}

We will now present our algorithm for maximally partially deciding whether a point escapes a closed set under the iteration of a continuous function.
More precisely, we consider the following decision problem:

\begin{definition}
    The \emph{Point Escape Problem} asks to decide for a given $d \in \N$, 
    a continuous function 
    $f \colon \R^d \to \R^d$, given via a $\left[\rho^d \to \rho^d\right]$-name,
    a closed set $A \subseteq \R^d$, given via an $\alpha^d$-name, and a given point 
    $x_0 \in A$, given via a $\rho^d$-name, 
    whether there exists a natural number $n \in \N$ with $f^{n}(x) \notin A$.

    An instance $(f,A,x_0)$ of the Point Escape Problem where $f^{n}(x) \notin A$ for some $n \in \N$ is called an \emph{escaping instance}.
    An instance $(f,A,x_0)$ where $f^{n}(x) \in A$ for all $n$ is called a \emph{trapped instance}.
\end{definition}

\begin{algorithm}[H]
    \caption{Point Escape Algorithm}\label{Algorithm: main algorithm}
    \begin{algorithmic}[1]
    \Procedure{Point-Escape}{$(F_n)_n, ((I^A_n, E^A_n)_n)_n, (X_n)_n$}
    \For{$n = 0, 1, \dots$}
        \Switch{\textsc{Point-Escape-Stage}($n, F_n, I^A_n, E^A_n, X_n$)} 
            \Case{\texttt{Escapes} $\Rightarrow$ \textbf{return} \texttt{Escapes}.}
            \EndCase
            \Case{\texttt{Trapped} $\Rightarrow$  \textbf{return} \texttt{Trapped}.}
            \EndCase
            \Case{\texttt{Unknown} $\Rightarrow$  \textbf{continue}.}
            \EndCase
        \EndSwitch
    \EndFor
    \EndProcedure
    \end{algorithmic}
\end{algorithm}
\begin{algorithm}[H]
    \caption{Point Escape Algorithm -- $n^{\text{th}}$ Stage}\label{Algorithm: main algorithm nth-stage}
    \begin{algorithmic}[1]
    \Procedure{Point-Escape-Stage}{$n \in \N, F \colon \QD^d_n\left(\left[-2^n, 2^n\right]^d\right) \rightrightarrows \QD^d_n, I^A \subf \QD^d_n, E^A \subf \QD^d_n, X \subf \QD^d_n$}
    \State Let 
            $
                C_n = \QD^d_n \setminus \QD^d_n\left(\left[-2^n, 2^n\right]^d\right).
            $
    \If {$|X| \cap |C_n| \neq \emptyset$}           \Comment{If $x_0$ is not certainly contained in $\left[-2^n, 2^n\right]^d$:}
        \State \textbf{return} $\texttt{Unknown}$.  \Comment{\hspace{1cm}proceed to the next stage.}
    \EndIf
    \State Let $Q_{n, 1} = F(X)$. 
    \algorithmiccomment{The set $Q_{n,i}$ will contain the point $f^{i}(x_0)$.}
    \State Let 
    $
        P_{n, 1} = \Set{Q \in Q_{n, 1}}
                  {
                    Q \not\Subset \left|E^A_n\right|
                  }.
    $
    \algorithmiccomment{The set $P_{n,1}$ will contain the point $x_0$.}
    \State Let $O_{n, 1} = Q_{n, 1}$.
    \algorithmiccomment{The set $O_{n,i}$ will contain the set}
    \State\algorithmiccomment{\hspace{1cm}$\{f(x_0), f^2(x_0),\dots,f^{i}(x_0)\}$.}
    \If {$Q_{n, 1} \cap C_n \neq \emptyset$}   \algorithmiccomment{If $F_n$ cannot be applied to $Q_{n, 1}$:}
    \State \textbf{return} $\texttt{Unknown}$. \algorithmiccomment{\hspace{1cm}proceed to the next stage.}
    \EndIf
    \For {$i = 1, 2, \dots$}
    \State Compute $Q_{n, i + 1} = F(Q_{n, i})$. \algorithmiccomment{Compute an inclusion of the next point on the orbit.}
    \State Compute 
    $
        P_{n, i + 1} = \Set{Q \in F(P_{n, i})}
                        {Q \not\Subset \left|E^A_n\right|}.
    $                                               \algorithmiccomment{The set $P_{n,i+1}$ contains $f(P_{n,i}) \cap A$.}
    \State Compute $O_{n, i + 1} = O_{n, i} \cup Q_{n, i + 1}$. 
    \algorithmiccomment{Compute the next inclusion of the whole orbit.}
    \If {$P_{n, i + 1} = \emptyset$}
    \algorithmiccomment{If $P_{n, i}$ is empty:}
    \State \textbf{return} $\texttt{Escapes}$.
    \algorithmiccomment{\hspace{1cm}the point must have escaped.}
    \ElsIf{$Q_{n, i + 1} \cap C_n \neq \emptyset$} \algorithmiccomment{If $F_n$ cannot be applied to $Q_{n, i}$:}
    \State \textbf{return} $\texttt{Unknown}$. \algorithmiccomment{\hspace{1cm}proceed to the next stage.}
    \ElsIf {$O_{n, i} = O_{n, i + 1}$}  \algorithmiccomment{If the orbit-inclusion is invariant under $F_n$:}
    \If{$O_{n, i} \Subset I^A_n$}       \algorithmiccomment{\hspace{1cm}If the orbit is certainly contained in $A$:}
    \State \textbf{return} $\texttt{Trapped}$. \algorithmiccomment{\hspace{1cm}\hspace{1cm}the instance must be trapped.}
    \Else                                      \algorithmiccomment{\hspace{1cm}Otherwise, no further progress can be made:}
    \State \textbf{return} $\texttt{Unknown}$. \algorithmiccomment{\hspace{1cm}\hspace{1cm}proceed to the next stage.}
    \EndIf 
    \EndIf
    \EndFor
    \EndProcedure
    \end{algorithmic}
\end{algorithm}

Algorithm \ref{Algorithm: main algorithm} above takes as input a name $(F_n)_n$ of a continuous function $f \colon \R^d \to \R^d$,
a name $(I^A_n, E^A_n)_n$ of a closed set $A \subseteq \R^d$,
and a name $(X_n)_n$ of a point $x_0 \in A$.
It either runs forever or halts within finitely many steps.
Upon halting, it reports whether there exists $n \in \N$ such that $f^n(x_0) \notin A$.
The algorithm proceeds in stages, working on a subdivision of $[-2^n, 2^n]^d$ into cubes of side-length $2^{-n}$ in each stage. 
The stages are described in Algorithm \ref{Algorithm: main algorithm nth-stage}.
At each stage, Algorithm \ref{Algorithm: main algorithm nth-stage} computes an overapproximation of the orbit of $x_0$ to determine whether $x_0$ escapes. 
At the same time, it searches for an invariant set under $F_n$ which is contained in $\intr{A}$, yielding an invariant set under $f$ which is contained in $\intr{A}$.

It is straightforward to verify that Algorithm \ref{Algorithm: main algorithm} is correct, \textit{i.e.}, that it produces the correct answer upon halting. 
\begin{proposition}\label{Proposition: correctness}
    Let $A \subseteq \R^d$ be a closed set, let $f \colon \R^d \to \R^d$ be a continuous function,
    and let $x_0 \in A$ be a point.
    Let $(I^A_n, E^A_n)_n$, $(F_n)_n$, and $(X_n)_n$ be names of $A$, $f$, and $x_0$ respectively.
    Assume that Algorithm \ref{Algorithm: main algorithm} is given the inputs $(I^A_n, E^A_n)_n$, $(F_n)_n$, and $(X_n)_n$.
    
    If the algorithm eventually halts, it correctly reports whether $x_0$ is trapped or whether it escapes.
\end{proposition}

\begin{proof}
    Assume that the algorithm halts and reports that the point escapes.
    Observe that for all $n$, the sequence $P_{n, 1}, P_{n, 2}, P_{n, 3}, \dots$ satisfies $f(x_0) \in |P_{n, 1}|$ and 
    $
        f\left(|P_{n, i}|\right) \cap A \subseteq \left|P_{n, i + 1}\right|.
    $
    Indeed, by definition we have 
    $
        P_{n, i + 1} = 
        \Set{Q \in F_n(P_{n, i})}
            {Q \not \Subset \left|E_n^A\right|}.
    $
    Let $y \in f(|P_{n, i}|) \cap A$.
    Then $y \in |F_n(P_{n, i})|$, so that $y \in Q$ for some $Q \in F_n(P_{n, i})$.
    If $Q \Subset \left|E_n^A\right|$, then $Q \subseteq \R^d \setminus A$, contradicting
    the assumption that $y \in A$.
    Hence, $Q \in P_{n, i + 1}$, so that $y \in |P_{n, i + 1}|$.

    Now, if the algorithm halts and reports that $x_0$ escapes under $f$, it has computed a finite sequence 
    $P_{n, 1}, P_{n, 2}, P_{n, 3}, \dots, P_{n, N}$
    satisfying 
    $f(x_0) \in |P_{n, 1}|$,
    $f\left(|P_{n, i}|\right) \cap A \subseteq \left|P_{n, i + 1}\right|$,
    and $P_{n, N} = \emptyset$.
    It follows that $f^k(x_0) \notin A$ for some $k \leq N$.

    Now assume that the algorithm halts and reports that the point is trapped.
    Observe that by construction we have 
    $O_{n, i + 1} = O_{n, i} \cup F_n(O_{n, i})$,
    so that 
    $F_n(O_{n, i}) \subseteq O_{n, i + 1}$.
    In particular, if $O_{n, i} = O_{n, i + 1}$, then $F(O_{n, i}) \subseteq O_{n, i}$.

    Hence, if the algorithm halts and reports that $x_0$ is trapped in $A$ under $f$,
    then the algorithm has computed a set of cubes $O_{n, i}$ with $f(x_0) \in O_{n, i}$,
    $f(|O_{n, i}|) \subseteq |F(O_{n, i})| \subseteq |O_{n, i}|$,
    and $|O_{n, i}| \Subset A$.
    It follows that $f^k(x_0) \in |O_{n, i}| \subseteq A$ for all $k \geq 1$.
\end{proof}

\subsection{Completeness}

We will now show that our algorithm is \emph{complete}, \textit{i.e.} that it halts on all problem instances for which the answer is robust under small perturbations.

It is easy to see that the algorithm halts on all problem instances where the point escapes.
We will show that if the point is trapped, termination of the algorithm is guaranteed by the existence of a \emph{robust invariant}.
Conversely, the absence of a robust invariant entails that the point escapes under arbitrarily small perturbations of the function.

\begin{definition}
Let $f \colon \R^d \to \R^d$ be a continuous function.
A \emph{robust invariant} for $f$ 
is a compact set $V \subseteq \R^d$ such that $f(V) \subseteq \intr{V}$.
\end{definition}

We first observe the following sufficient conditions for termination:
\begin{proposition}\label{Proposition: termination conditions}
    Let $A \subseteq \R^d$ be a closed set, let $f \colon \R^d \to \R^d$ be a continuous function,
    and let $x_0 \in A$ be a point.
    Let $(I^A_n, E^A_n)_n$, $(F_n)_n$, and $(X_n)_n$ be names of $A$, $f$, and $x_0$ respectively.
    Assume that Algorithm \ref{Algorithm: main algorithm} is given the inputs $(I^A_n, E^A_n)_n$, $(F_n)_n$, and $(X_n)_n$.
    \begin{enumerate}
        \item If $x_0$ escapes $A$ under $f$, then the algorithm eventually halts.
        \item If $f(x_0)$ is contained in the interior of a robust invariant $V$ for $f$ with $V \subseteq \intr{A}$, then the algorithm eventually halts.
    \end{enumerate}
\end{proposition}

\begin{proof}
    \begin{enumerate}
    \item 
        For $n \in \N$, define $Q_{n, 1} = F_n(X_n)$, 
        $P_{n, 1} = \Set{Q \in Q_{n, 1}}{Q \not \Subset |E_n^A|}$.
        For $i \geq 1$, define $Q_{n, i}$ and $P_{n, i}$ recursively via 
        $Q_{n, i + 1} = F_n(Q_{n, i})$,
        and
        $P_{n, i + 1} = \Set{Q \in F_n(P_{n,i})}{Q \not \Subset |E_n^A|}$.
        The double-sequences $(P_{n, i})_{n, i}$ and $(Q_{n, i})_{n,i}$ thus defined agree with the 
        double-sequences $(P_{n, i})_{n, i}$ and $(Q_{n, i})_{n,i}$ that the algorithm computes, 
        except that the algorithm does not necessarily compute the sequence elements for all values of $n$ and $i$.

        Assume that the point eventually enters $\R^d \setminus A$.
        Then $f^N(x_0) \in \R^d \setminus A$ for some minimal $N$.
        The sequence $x_0, f(x_0), \dots, f^N(x_0)$ is bounded, say by $B \in \N$.
        By Lemma \ref{lemma: modulus of continuity witnessed by representation} there exists a function $\omega \colon \N \to \N$ such that 
        \[
            Q \in \QD^d_{\omega(n)}([-B,B]^d) \rightarrow \operatorname{diam} \left|F_{\omega(n)}(Q)\right| < 2^{-n}.
        \]
        We may assume that $\omega(n) \geq n$ for all $n \in \N$.
        Since $\R^d \setminus A$ is open, there exists $m_0 \in \N$ such that 
        $B(f^N(x_0), 2^{-m_0}) \subseteq \R^d \setminus A$.
        We may assume that $m_0$ is chosen so large that $|C_{m_0}|$ is contained in the complement of $[-B - 1,B + 1]^d$.

        Since $(X_n)_n$ is a name of $X_0$, there exists $m_1$ such that $\operatorname{diam} |X_{m_1}| < 2^{-\omega^N(m_0 + 1)}$. 

        Let $n_0 = \max\{m_0 + 1, m_1\}$.
        We claim that $|Q_{n_0, N}|$ is contained in $\R^d \setminus A$ and that $P_{n_0, i}$ does not intersect $C_{n_0}$ for $i \leq N$.
        
        By definition, $Q_{n_0, 1} = F_{n_0}(X_{n_0})  \subseteq \QD^d_{n_0}$ contains $x_0$ and satisfies 
        $\operatorname{diam} |Q_{n_0, 0}| < 2^{-\omega^{N - 1}(m_0 + 1)}$. 
        We have $Q_{n_0, i} = F_{n_0}(Q_{n_0, i - 1})$ for $i = 2, 3, \dots$.
        An easy induction shows that if $i \leq N$, then $Q_{n_0, i}$ is contained in $[-B - 1, B + 1]^d$ and that 
        $\operatorname{diam} |Q_{n_0, i}| < 2^{-\omega^{N - i}(m_0 + 1)}$.
        In particular, $P_{n_0, i} \subseteq Q_{n_0, i}$ does not intersect $C_n$ for $i \leq N$.
        We obtain $\operatorname{diam} |Q_{n_0, N}| < 2^{-m_0 - 1}$.
        By construction, $|Q_{n_0,N}|$ contains the point $f^N(x_0)$.
        We obtain
        \[ 
            |Q_{n_0, N}|
            \subseteq
            B(f^N(x_0), 2^{-m_0}) 
            \subseteq
            \R^d \setminus A.
        \]

        Since $|P_{n, N}| \subseteq |Q_{n, N}|$ for all $n$, we obtain $|P_{n_0, N}| \subseteq \R^d \setminus A$.

        Since $|P_{n + 1, N}| \subseteq |P_{n, N}|$, and $|P_{n, N}|$ is compact for all $n$,
        it follows from Proposition \ref{Proposition: witnessing containment of compact sets} that 
        $|P_{n, N}| \subseteq |P_{n_0, N}| \Subset |E^A_n|$ for all sufficiently large $n$.
        Since $P_{n, N}$ is disjoint from $E^A_n$ by construction, it follows that $P_{n, N} = \emptyset$.
        Further, since $P_{n, i}$ does not intersect $C_{n}$ for $n \geq n_0$ and $i \leq N$, 
        if $n$ is sufficiently large and the algorithm has not halted 
        within $n$ stages, 
        the algorithm will not return $\texttt{Unknown}$ during the $n^{\text{th}}$ stage and halt at the latest upon checking 
        whether $P_{n, N}$ is empty,
        reporting that the point escapes.
    \item 
        Assume that $f$ has a robust invariant $V \subseteq \intr{A}$ with $f(x_0) \in \intr{V}$.
        For $n \in \N$, let $V_n$ denote the set of all cubes $Q \in \QD^d_n$ with $|Q| \subseteq \intr{V}$.
        We claim that for sufficiently large $m$ we have $F_m(V_m) \subseteq V_m$.
        Indeed, we have by assumption that $f(V) \subseteq \intr{V}$.
        Since $f(V)$ is compact, there exists $n \in \N$ such that for all $y \in f(V)$ we have $B(y, 2^{-n}) \subseteq \intr{V}$.
        By Lemma \ref{lemma: modulus of continuity witnessed by representation} there exists $m \in \N$ such that 
        $Q \in \QD^d_m(V)$ implies $\operatorname{diam} |F(Q)| < 2^{-n-1}$.
        Letting $x \in Q\in \QD^d_m(V)$, we have $f(x) \in |F(Q)|$, which implies 
        $
            \intr{V} \supseteq B(f(x), 2^{-n}) \supseteq |F(Q)|
        $.
        The claim follows.

        We now claim that the algorithm halts eventually.
        Since $f(x_0) \in \intr{V}$, we have $Q_{n, 1} \in V_n$ for sufficiently large $n$.
        We may assume that $n$ is so large that $F_n(V_n) \subseteq V_n$.
        It then follows that $Q_{n, i} \subseteq V_n$ for all $i$ and thus 
        $O_{n, i} \subseteq V_n$ for all $i$.
        Since $V$ is compact, the set $V_n$ is finite.
        The sequence $O_{n, 1} \subseteq O_{n, 2} \subseteq \dots$ is strictly increasing and bounded by $V_n$.
        Hence, we must eventually have $O_{n, i} = O_{n, i + 1}$.
        Algorithm \ref{Algorithm: main algorithm nth-stage} now halts and returns $\texttt{Trapped}$ provided that $O_{n, i} \Subset I^A_n$, and otherwise returns $\texttt{Unknown}$.
        Thus, the algorithm computes a sequence of sets of cubes $O_{N, i_N}, O_{N, i_{N + 1}}$ with 
        $F_n(O_{n, i_n}) \subseteq O_{n, i_n}$ and $|O_{n, i_n}| \subseteq \intr{V} \subseteq \intr{A}$ and checks if 
        $O_{n, i_n} \Subset I^A_n$.
        Since $F_{n + 1}$ refines $F_n$ in the sense that $Q' \subseteq Q$ implies
        $|F_{n + 1}(Q')| \subseteq |F_n(Q)|$, we have $|O_{n, i_n}| \supseteq |O_{n + 1, i_{n + 1}}|$ for all $n$.
        It follows from Proposition \ref{Proposition: witnessing containment of compact sets} that for all large $n$ we have 
        $I^A_n \Supset \left|O_{N, i_N}\right| \supseteq \left|O_{n, i_n}\right|$,
        so that Algorithm \ref{Algorithm: main algorithm nth-stage} eventually outputs $\texttt{Trapped}$.
        This finishes the proof.
\end{enumerate}
\end{proof}

The core of the completeness proof is the following
\begin{lemma}\label{Lemma: main perturbation theorem}
    Let $f \colon \R^d \to \R^d$ be a continuous map.
    Let $K \subseteq \R^d$ be a compact set.
    Let $x_0 \in \intr{K}$.

    Assume that $x_0$ is not contained in the interior of any robust invariant of $f$ in $K$.
    Then for all $\varepsilon > 0$ there exists a map 
    $\widetilde{f} \colon \R^d \to \R^d$
    with
    $
        \norm{\widetilde{f} - f}_{\infty} < \varepsilon
    $ 
    such that $x_0$ escapes $K$ under $\widetilde{f}$.
\end{lemma}
\begin{proof}[Proof Sketch]
    We will only sketch the proof idea.
    The full proof is given in Appendix \ref{Appendix: proof of main perturbation theorem}.

    Let $\varepsilon > 0$.
    Consider the set 
    $
        E_0 = \R^d \setminus \intr{K}.
    $
    Since $K$ is compact and $f$ is continuous,
    there exists $0 < \delta < \varepsilon/4$ such that 
    $x_0 \notin \overline{B}(E_0, \delta)$
    and such that $x, y \in K$,
    $d(x,y) \leq \delta$ implies $d(f(x), f(y)) < \varepsilon/4$.
    
    Consider the compact set $K_0 = K \setminus B(E_0, \delta)$.
    Then by construction, $K_0 \Subset K$ and $x_0$ is contained in the interior of the complement of $K_0$.
    By assumption, $K_0$ cannot be a robust invariant of $f$.
    Hence, the set 
    $
        E_{1}
            =
            \Set{ x \in K \setminus B(E_0, \delta) }
                { 
                    f(x) \in \overline{B}(E_0,\delta) 
                }
    $
    must be non-empty.
    If $x_0 \notin \overline{B}(E_1, \delta)$, we can repeat the same argument, 
    applied to 
    $K_1 = K \setminus \left(B(E_0, \delta) \cup B(E_1, \delta)\right)$
    to obtain the non-empty closed set
    $
        E_{2}
            =
            \Set{ x \in K \setminus \left(B(E_0, \delta) \cup B(E_1, \delta)\right)}
                { 
                    f(x) \in \overline{B}(E_1,\delta) 
                }.
    $

    Proceeding by induction, we obtain a sequence of non-empty closed subsets $E_1, E_2, \dots$
    of $K$ 
    such that $E_i$ is disjoint from $\bigcup_{j = 0}^{i - 1} B(E_j, \delta)$
    and for  
    $x \in K \setminus \bigcup_{j = 0}^{i}B(E_j, \delta)$
    we have $f(x) \in \overline{B}(E_{i}, \delta)$ if and only if $x \in E_{i + 1}$.
    A straightforward compactness argument shows that the sequence must be finite, so that 
    there exist $N \geq 1$ with $x_0 \in \overline{B}(E_N, \delta)$
    and $x_0 \notin \overline{B}(E_i, \delta)$ for $i < N$.  

    Now, since $x_0$ is $\delta$-close to an element of $E_N$, which gets mapped $\delta$-close to some element of $E_{N - 1}$ under $f$, we can modify $f$ on a small neighbourhood of $x_0$ such that $x_0$ gets mapped to a point $x_1 \in E_{N - 1}$.
    We then proceed to modify $f$ on a small neighbourhood of $x_1$ to ensure that $x_1$ maps to a point $x_2 \in E_{N - 2}$.
    By induction, we obtain a sequence of points $x_1, \dots, x_N$ and a perturbation $\widetilde{f}$ of $f$ by at most $\varepsilon$, such that 
    $\widetilde{f}(x_i) = x_{i + 1}$ and $x_N \notin K$.
\end{proof}

Lemma \ref{Lemma: main perturbation theorem} applies only to compact sets.
It admits the following extension to arbitrary closed sets:

\begin{lemma}\label{Lemma: extension to arbitrary closed sets}
    Let $f \colon \R^d \to \R^d$, let $A \subseteq \R^d$ be a closed set,
    and let $x_0 \in A$.

    If there is no robust invariant $V$ for $f$ with 
    $f(x_0) \in \intr{V} \subseteq V \subseteq \intr{A}$, then there exist
    sequences $(f_n)_n$ and $(A_n)_n$ with $f_n \to f$ and $A_n \to A$ 
    such that
    $x_0 \in A$ 
    and
    $x_0$ escapes $A_n$ under $f_n$.

    Further, if $A \neq \R^d$, then we may take $A_n = A$ for all $n$.
\end{lemma}
\begin{proof}
    Fix a name $\left((I^A_n, E^A_n)\right)_n$ of $A$.
    For all $n \in \N$, the prefix 
    $(I^A_1, E^A_1), \dots, (I^A_n, E^A_n)$
    extends to a name of a compact set $K_n \subseteq A$
    with $x_0 \in K_n$.
    By assumption, there is no robust invariant of $f$ in $\intr{A}$ that contains $f(x_0)$.
    Hence, there is no such invariant in $K_n$.
    By Lemma \ref{Lemma: main perturbation theorem}, there exists a continuous function
    $f_n$
    with $\norm{f - f_n}_{\infty} < 2^{-n}$
    such that $x_0$ escapes $K_n$ under $f_n$. 
    By construction, we have $K_n \to A$ and $f_n \to f$.

    Now, assume that $A \neq \R^d$.
    Then the set $\R^d \setminus A$ is open and non-empty.
    Fix a point $y \in \R^d \setminus A$.

    Consider the sets $C_n = [-2^n, 2^n]^d \cap A$.
    Since there is no robust invariant of $f$ contained in the interior of $A$, there is no such invariant contained in the 
    interior of $C_n$.
    Hence, there exists a function $g_n \colon \R^d \to \R^d$ with 
    $\norm{g_n - f}_{\infty} < 2^{-n}$
    such that $x_0$ escapes $C_n$ under $g_n$.
    Let $k$ be the smallest integer with $g_n^k(x_0) \notin C_n$.
    We distinguish two cases:
        If $g_n^k(x_0) \notin A$, then $x_0$ escapes $A$ under $g_n$.
        In that case, let $f_n = g_n$.

        If $g_n^k(x_0) \in A$, then $g_n^k(x_0) \notin [-2^n, 2^n]^d$.
        Choose $\varepsilon > 0$ so small that $B(g_n^k(x_0), \varepsilon) \cap [-2^n, 2^n]^d = \emptyset$.
        Define 
        \[
            f_n(x) = 
                \begin{cases}
                    \left(1 - \frac{\norm{x - g_n^k(x_0)}}{\varepsilon}\right) y
                    + g_n(x) \frac{\norm{x - g_n^k(x_0)}}{\varepsilon} &\text{if }x \in B(g_n^k(x_0), \varepsilon),\\
                    g_n(x) &\text{otherwise.}
                \end{cases}
        \]
        Then, by construction, $f_n^{k + 1}(x_0) = y \notin A$.
        It follows from 
        $\norm{f_n - f}_{\infty, C_n} = \norm{g_n - f}_{\infty, C_n} < 2^{-n}$
        that $f_n \to f$ as $n \to \infty$.
\end{proof}

Our main theorem now follows immediately from 
Proposition \ref{Proposition: correctness},
Proposition \ref{Proposition: termination conditions},
and Lemma \ref{Lemma: extension to arbitrary closed sets}:

\begin{theorem}\label{Theorem: main theorem}
    Let $A \subseteq \R^d$ be a closed set, let $f \colon \R^d \to \R^d$ be a continuous function,
    and let $x_0 \in A$ be a point.
    Let $(I^A_n, E^A_n)_n$, $(F_n)_n$, and $(X_n)_n$ be names of $A$, $f$, and $x$ respectively.
    Assume that Algorithm \ref{Algorithm: main algorithm} is given the inputs $(I^A_n, E^A_n)_n$, $(F_n)_n$, and $(X_n)_n$.
    \begin{enumerate}
        \item If the algorithm eventually halts, it correctly reports whether $x_0$ is trapped or whether it escapes.
        \item If $x_0$ escapes $A$ under $f$, then the algorithm eventually halts.
        \item If $f(x_0)$ is contained in the interior of a robust invariant $V$ of $f$ with $I \subseteq \intr{A}$, then the algorithm eventually halts.
        \item If the algorithm does not halt, then there exist sequences $(f_{0, n})_n$,  $(f_{1, n})_n$, $(A_n)_n$  
        with $f_{i,n} \to f_n$ for $i = 0,1$ and $A_n \to A$, 
        such that $x_0$ escapes $A_{0,n}$ under $f_{0, n}$ and is trapped in $A_{1,n}$ under $f_{1, n}$ for all $n$.
        Moreover, if $A \neq \R^d$, then we can choose $A_{0,n} = A_{1,n} = A$ for all $n$.
    \end{enumerate}
    In particular, Algorithm \ref{Algorithm: main algorithm} is a complete decision method for the Point Escape Problem.
\end{theorem}

Theorem \ref{Theorem: main theorem} establishes more than just the completeness of the algorithm: if $A$ is not equal to all of $\R^d$ and the algorithm fails to halt on a problem instance $(f, A, x_0)$, then the answer to the problem is unstable under small perturbations of the function $f$ alone. 
This implies that Algorithm \ref{Algorithm: main algorithm} yields a complete decision method for the Special Point Escape Problem where $A$ and $x_0$ are fixed, and only $f$ is given as an input.
Of course, this result does not extend to the case where $A = \R^d$ since every point is robustly trapped in $\R^d$, but not every function has a robust invariant.

\subsection{Generic Termination}

A complete algorithm for a decision problem need not be able to solve the problem on a ``large'' set of instances.
For example, the set $\Q \subseteq \R$ is maximally partially decided by the algorithm that never halts.
Given a complete algorithm, one is thus led to the problem of characterising the ``size'' of its halting set.

All escaping instances of the Point Escape Problem are robustly escaping -- in particular, Algorithm \ref{Algorithm: main algorithm} halts on a dense set of inputs -- which means that for every input $(f,A,x_0)$ there exists a sequence of inputs $(f_n,A_n,x_{0,n})$ with $f_n \to f$, $A_n \to A$, $x_{0,n} \to x_0$ such that the algorithm halts on input $(f_n, A_n, x_{0,n})$ for all $n$. Convergence is with respect to the final topology induced by the representation.
It is of course not the case that all trapped instances are robustly trapped. 
However, we will show that the robustly trapped instances are typical among the trapped ones: every boundary instance (\textit{i.e.} every trapped instance that is not robustly trapped) can be perturbed into a robustly trapped instance under arbitrarily small perturbations.
In this sense, both the robustly escaping instances and the robustly trapped instances constitute ``large'' sets.

\begin{lemma}\label{Lemma: generic termination of point-escape algorithm}
    Let $(f, K, x_0)$ be a boundary instance of the Point Escape Problem, where $K$ is a compact set.
    Then there exist sequences $(f_n)_n$ and $(K_n)_n$ with $f_n \to f$, $K_n \to K$ such that
    $(f_n, K_n, x_0)$ is a robust trapped instance of the Point Escape Problem for all $n$.
\end{lemma}
\begin{proof}[Proof Sketch.]
    Since Algorithm \ref{Algorithm: main algorithm} halts on all escaping instances, the instance $(f, K, x_0)$ must be trapped.

    Let $\varepsilon > 0$.
    Consider the compact set $K_{\varepsilon} = \bigcup_{y \in K}\overline{B}(y,\varepsilon)$.
    Then the orbit $\Set{f^{n}(x_0)}{n \in \N}$ is contained in the interior of $K_\varepsilon$.

    By compactness, $K_{\varepsilon}$ admits a cover of the form 
    $\bigcup_{i = 0}^N  B(y_i, \varepsilon)$
    with $y_i \in K_{\varepsilon}$.
    The orbit $\Set{f^{n}(x_0)}{n \in \N}$ is contained in $K_{\varepsilon}$, so that, by the pigeonhole principle, there exist minimal $j < k \leq N + 1$ such that 
    $f^{j}(x_0)$ and $f^{k}(x_0)$ are contained in the same ball $B(y_i, \varepsilon)$.
    The points $f^j(x_0), f^{j + 1}(x_0), \dots, f^{k - 1}(x_0)$ are uniformly bounded away from the boundary of $K_{\varepsilon}$ by some $\delta > 0$.
    We may further assume that the balls $\overline{B}(f^i(x_0), \delta)$ are disjoint for $i \in \{j, \dots, k\}$ and that $\delta$ is so small that $d(x,y) < \delta$ implies $d(f(x), f(y)) < \varepsilon$.
    For $j \leq i < k - 1$ we replace $f$ on $B(f^i(x_0), \delta/2)$ by the constant function with value $f^{i + 1}(x_0)$.
    On $B(f^{k - 1}(x_0), \delta/2)$, we replace $f$ by the constant function with value $f^{j}(x_0)$.
    To make the function continuous again, we use the annuli
    $B(f^{i}(x_0), \delta) \setminus B(f^{i}(x_0), \delta/2)$
    to connect $f$ with the modified function via linear interpolation.
    Then the set 
    $\bigcup_{i = j}^{k - 1} \overline{B}(f^{i}(x_0), \delta/2)$
    is a robust invariant for the perturbed function which is contained in the interior of $K_{\varepsilon}$ and contains $f^j(x_0)$ in its interior.
    From this, we easily obtain a robust invariant in the interior of $K_{\varepsilon}$ that contains $x_0$ in its interior.

    The full proof is given in Appendix \ref{Appendix: proof of generic termination of point-escape algorithm}.
\end{proof}

\begin{corollary}
    Let $(f, A, x_0)$ be a boundary instance of the Point Escape Problem.
    Then there exist sequences $(f_n)_n$ and $(A_n)_n$ with $f_n \to f$, $A_n \to A$ such that
    $(f_n, A_n, x_0)$ is a robust trapped instance of the Point Escape Problem for all $n$.
\end{corollary}
\begin{proof}
    The instance $(f,A,x_0)$ must be a trapped instance of the Point Escape Problem.
    We show that there exists a sequence of functions $(f_n)_n$ and a sequence of compact sets $(K_n)_n$
    with $f_n \to f$ and $K_n \to A$ such that $(f_n, K_n, x_0)$ is a (not necessarily robust) trapped instance 
    of the Point Escape Problem.
    Together with Lemma \ref{Lemma: generic termination of point-escape algorithm} this yields the claim.

    We distinguish two cases: If the orbit of $x_0$ under $f$ is bounded, say, contained in the ball $\overline{B}(0, N)$ for some $N$, 
    we let $f_n = f$ and $K_n = \overline{B}(0, N + 2^n)$ for all $n$. It is clear that $x_0$ is trapped in $K_n$ under $f$ for all $n$. It is easy to see that $K_n \to A$ in the topology induced by our representation of closed sets.

    Now assume that the orbit of $x_0$ under $f$ is unbounded.
    Fix some $N > \norm{x}_{\infty}$.
    Let $x_j = f^{(j)}(x_0)$.
    Let $n \in \N$.
    Let $m$ be the smallest index such that $x_0, \dots, x_m$ are contained in $\overline{B}(0, N +  2^n)$, but 
    $x_{m + 1} \notin \overline{B}(0, N + 2^n)$.
    Let $\delta = \norm{x_{m + 1}}_\infty - N - 2^n > 0$.
    Let $K_n = A \cap \overline{B}(0, N + 2^n +\delta)$.
    Let 
    \[
        f_n(x) = 
        \begin{cases}
            f(x)                                  &\text{if }x \in \overline{B}(0, N +  2^n),\\
            \frac{\norm{x}_{\infty} - N -  2^n}{\delta} x_0 + \left(1 - \frac{\norm{x}_{\infty} - N -  2^n}{\delta}\right)f(x)   &\text{if }x \in \overline{B}(0, N +  2^n + \delta) \setminus \overline{B}(0, N +  2^n),\\
            x_0                                   &\text{if }x \in \R^d\setminus\overline{B}(0, N +  2^n + \delta).
        \end{cases}
    \]
    Observe that by construction we have $f_n(x_{m + 1}) = x_0$, so that the orbit of $x_0$ under $f_n$ is equal to $\{x_0,\dots,x_{m + 1}\} \subseteq K_n$.
    It is easy to see (for example by using Proposition \ref{Proposition: properties of function-space representation}) that $f_n \to f$ in the topology induced by our representation of functions and that $K_n \to K$ in the topology induced by our representation of closed sets.
\end{proof}

\input{Linear-Systems}

\input{quadratic-short}

\section{Conclusion and Future Work}

We have studied the problem of deciding for a given set, initial point, and function whether the initial point escapes the set under iteration of the function. We have allowed arbitrary continuous functions as inputs, represented by the weakest representation that makes function evaluation uniformly computable. We have given an algorithm that certifies that the initial point is trapped by searching for a robust invariant.
We have shown that this algorithm is complete in the sense that no sound algorithm can detect further trapped instances.
This settles the very natural question of how much is decidable about the escape problem when we are given just enough information about our system to compute the orbit of a given point.
In general, computable analysis allows us to ask and answer questions of this form: if we are given a certain (limited) amount of information about a system, how much can we say about the system's behaviour?
This leads to various natural directions for future work: 
On the one hand, it would be very interesting to study more general systems, for example by replacing $\R^d$ with a locally compact metric space $X$.
On the other hand, it would be interesting to study more special systems which are presented with additional information, such as polynomials or other subclasses of analytic maps which are ``well-behaved at infinity'' such as Pfaffian functions.
We have provided partial evidence that, somewhat surprisingly, our algorithm for general systems yields maximal partial algorithms for rather rigid special systems, such as linear systems and -- conjecturally -- one-dimensional quadratic systems.
It would be very interesting to determine whether this reduction extends, for example, to higher-degree multivariate polynomials.

\bibliography{escape}
\newpage

\appendix

\section*{Appendix}

\section{Proofs for Section \ref{Section: Encodings}}\label{Appendix: Representations}

\begin{proposition}\label{Proposition: closed balls about closed sets are closure of open balls}
    Let $A \subseteq \R^d$ be a closed set.
    Then $\overline{B}(A, r)$ is the closure of $B(A, r)$.
\end{proposition}
\begin{proof}
    Assume that $y \in \overline{A}(K, r)$.
    Then $y \in \overline{B}(x, r)$ for some $x \in A$.
    Since closed balls in $\R^d$ are the closure of open balls,
    $y$ is contained in the closure of $B(x, r)$.
    Since $B(x, r)$ is contained in $B(A, r)$, the closure of $B(x, r)$ is 
    contained in the closure of $B(A, r)$.
    In particular, $y$ is contained in the closure of $B(A, r)$.

    Conversely, let $y$ be contained in the closure of $B(A, r)$.
    If $y$ is contained in the closure of $B(K, r)$, then it is the limit of a sequence $(y_n)_n$ with $y_n \in B(A, r)$.
    The sequence $\{\norm{y_0}_\infty, \norm{y_1}_\infty, \dots\}$ is bounded by some integer $B > 0$.
    Let $K = A \cap [-B - r, B + r]^d$.
    Then the sequence $(y_n)_n$ is contained in $B(K, r)$.
    Indeed, for all $n \in \N$ there exists $x \in A$ with $\norm{y_n - x}_{\infty} < r$.
    We have 
    \[
        \norm{x}_\infty \leq \norm{y_n}_\infty + \norm{y_n - x}_{\infty} < B + r
    \]
    so that $x \in K$.
    
    Since $K$ is compact and the map $x \mapsto \norm{y_n - x}_\infty$ is continuous for all $n$, there exists a sequence $(x_n)_n$ in $K$ such that 
    \[ 
        \norm{y_n - x_n}_\infty = \min \Set{\norm{y_n - x}_\infty}{x \in K} < r.
    \]
    By the Bolzano-Weierstraß theorem, the sequence $(x_n)_n$ has a convergent subsequence $(x_{n_k})_k$ with limit $x \in K$.
    We have 
    \[
        \norm{y_{n_k} - x}_\infty
        \leq 
        \norm{y_{n_k} - x_{n_k}}_\infty
        +
        \norm{x_{n_k} - x}_\infty
        <
        r + \norm{x_{n_k} - x}_\infty.
    \] 
    Taking the limit as $k \to \infty$ we obtain $\norm{y - x}_\infty \leq r$, so that 
    $y \in \overline{B}(K, r) \subseteq \overline{B}(A, r)$.
\end{proof}

\begin{proposition}\label{Proposition: openness and fibre-separableness}
    Let $\delta \colon \subseteq \NN \to X$ be a representation.
    Assume that there exists a continuous map 
    \[
        F \colon \dom (\delta) \to \left(\dom (\delta)\right)^\N
    \]
    such that for all $p \in \dom(\delta)$, the sequence 
    $F(p)$ is dense in $\delta^{-1}(\delta(p))$.
    Then $\delta$ is an open map.
\end{proposition}
\begin{proof}
    Let $U \in \O(\dom \delta)$ be an open set.
    Then the set 
    \[
        V = \Set{p \in \dom (\delta)}{ \exists n \in \N. F(p)_n \in U }
    \]
    is open.
    We claim that 
    \[
        V = \delta^{-1}(\delta(U))
    \]
    which shows that $\delta(U)$ is open.

    Let $p \in V$.
    Then there exists a $q \in \dom \delta$ with $\delta(q) = \delta(p)$ and 
    $q \in U$.
    We have $\delta(q) \in \delta(U)$ and hence 
    $p \in \delta^{-1}(\delta(p)) \subseteq \delta^{-1}(\delta(U))$.

    Conversely, let $p \in \delta^{-1}(\delta(U))$.
    Then $\delta(p) \in \delta(U)$, so that there exists 
    $q \in \dom \delta$ with $q \in U$ and $\delta(q) = \delta(p)$.
    Since $F(p)$ is dense in $\delta^{-1}(\delta(p))$, there exists a subsequence
    $(q_n)_n$ of $F(p)$ with $q_n \to q$.
    By definition, we must have $q_n \in U$ for all large $n$.
    This implies $p \in V$ by definition.
\end{proof}

\begin{proposition}\label{Proposition: openness and fibre-overtness}
    Let $\delta \colon \subseteq \NN \to X$ be a representation.
    Assume that there exists an algorithm that takes as input 
    the name $p \in \dom \delta$ of a point $x$ and 
    a finite sequence $q_0 \in \N^*$ of integers and 
    halts if and only if $q_0$ can be extended to a name $q \in \dom \delta$ of $x$.
    Then $\delta$ is an open map.
\end{proposition}
\begin{proof}
    Any finite sequence $q_0 \in \N^*$ defines an open set 
    \[
        \left[q_0\right] = \Set{q \in \dom \delta}{q(n) = q_0(n) \text{ for }n \leq \operatorname{length}(q_0)}.
    \]
    The family of sets $\left([q_0]\right)_{q_0 \in \N^*}$ form a basis for the topology on $\dom \delta$.
    It suffices to show that $\delta$ sends basic open sets to open sets.
    
    Let $\left[q_0\right]$ be a basic open set.
    By assumption, the set 
    \[
        V = \Set{p \in \dom \delta}{ q_0 \text{ has an extension }q \in \dom \delta \text{ with }\delta(q) = \delta(p) }
    \]
    is open.
    We claim that 
    $V = \delta^{-1}\left(\delta\left(\left[q_0\right]\right)\right)$.
    It then follows from the definition of the topology on $X$ that $\delta\left(\left[q_0\right]\right)$ is open.

    Let $p \in V$.
    Then there exists an extension $q$ of $q_0$ with $\delta(p) = \delta(q)$.
    Then $\delta(p) = \delta(q) \in \delta\left(\left[q_0\right]\right)$,
    so that $p \in \delta^{-1}\left(\delta\left(\left[q_0\right]\right)\right)$.
    
    Conversely, let $p \in \delta^{-1}\left(\delta\left(\left[q_0\right]\right)\right)$.
    Then $\delta(p) \in \delta\left(\left[q_0\right]\right)$,
    so that there exists $q \in \left[q_0\right]$ with $\delta(q) = \delta(p)$.
    By definition, $p \in V$.
\end{proof}

\begin{proposition}\label{Proposition: rho-d admissible and open (Appendix)}
    The representation $\rho^d$ is open and admissible and induces the standard Euclidean topology on $\R^d$,
    \textit{i.e.}, the topology generated by the norm $\norm{\cdot}_\infty$.
\end{proposition}
\begin{proof}
    For the claim of admissibility and the final topology, see \cite{WeihrauchBook,SchroederPhD}.
    The openness of the representation follows from Proposition \ref{Proposition: openness and fibre-overtness}:
    Given a name $(X_n)_n$ of a point $x \in \R^d$
    and a finite sequence $\langle Y_1, \dots, Y_N \rangle$ of finite sets of cubes,
    observe that $\langle Y_1, \dots, Y_N \rangle$ extends to a name 
    of $x$ if and only if $|Y_{n + 1}| \subseteq |Y_n|$ for all $n < N$
    and $|X_n| \Subset |Y_{N}|$ for some $n \in \N$.
    This is semi-decidable, since there are only finitely many cubes to check for each $n \in \N$.
\end{proof}

\begin{proposition}\label{Proposition: rho-d-rho-d admissible and open (Appendix)}
    The representation $\left[\rho^d \to \rho^d\right]$ is open and admissible and induces the compact-open topology on $C\left(\R^d, \R^d\right)$.
    This topology further coincides with the topology induced by the metric \eqref{eq: distance on functions}.
\end{proposition}
\begin{proof}
    For the claim of admissibility and the final topology, see \cite{WeihrauchBook,SchroederPhD}.
    The openness of the representation follows from Proposition \ref{Proposition: openness and fibre-overtness}:
    Given a name $(F_n)_n$ of a function $f \colon \R^d \to \R^d$
    and a finite sequence $\langle G_1, \dots, G_n \rangle$ of set-valued maps  
    $G_n \colon \QD^d_n([-2^n, 2^n]^d) \rightrightarrows \QD^d_n$
    with finite non-empty values, 
    observe that $\langle G_1, \dots, G_N \rangle$ extends to a name 
    of $f$ if and only if $|G_{n + 1}(Q')| \subseteq |G_n(Q)|$ 
    whenever $Q' \subseteq Q$
    for all $n < N$
    and there exists $n \in \N$ 
    such that for all $Q' \in \QD^d_n([-2^n, 2^n]^d)$
    and all $Q \in \QD^d_N([-2^N, 2^N]^d)$
    with $Q' \subseteq Q$
    we have 
    $|F_n(Q')| \Subset |G_N(Q)|$.
    This is semi-decidable, since there are only finitely many cubes to check for each $n \in \N$.
\end{proof}

\begin{proposition}\label{Proposition: alpha is admissible and open (Appendix)}
    The representation $\alpha^d$ is admissible and open.
\end{proposition}
\begin{proof}
    For the claim of admissibility, see \cite{WeihrauchBook,SchroederPhD}.
    
    The openness of the representation follows from Proposition \ref{Proposition: openness and fibre-overtness}:
    Given a name $(I^A_n, E^A_n)_n$ of a closed set $A \subseteq \R^d$
    and a finite sequence $\langle (J^A_1, F^A_1), \dots, (J^A_N, F^A_N) \rangle$ of pairs of finite sets of cubes,
    observe that $\langle (J^A_1, F^A_1), \dots, (J^A_N, F^A_N) \rangle$ extends to a name of $A$ if and only if
    for all $n < N$ if $Q \in J^A_n$ and $Q' \in \QD^d_{n + 1}$ with $Q' \subseteq Q$, then $Q' \in J^A_{n + 1}$,
    the same property holds true for $F^A_n$ and $F^A_{n + 1}$,
    and there exists some $n \in \N$ such that 
    $|J^A_N| \Subset |I^A_n|$
    and 
    $|F^A_N| \Subset |E^A_n|$.
    This is semi-decidable, since there are only finitely many cubes to check for each $n \in \N$.
\end{proof}

\section{Proof of Lemma \ref{Lemma: main perturbation theorem}}\label{Appendix: proof of main perturbation theorem}

Before we prove Lemma \ref{Lemma: main perturbation theorem}, we establish two simple auxiliary results.

\begin{lemma}\label{Lemma: disjoint sequence of sets in a compact set}
    Let $K$ be a compact set.
    Let $\delta > 0$.
    Then there does not exist an infinite sequence
    $(E_i)_i$ 
    of non-empty subsets of $K$ such that 
    $E_i$ is disjoint from $\bigcup_{j = 0}^{i - 1} B(E_j, \delta)$
    for all $i$.
\end{lemma}
\begin{proof}
    Suppose for the sake of contradiction that such a sequence exists.
    By total boundedness there exist finitely many points $x_1, \dots, x_N \in K$
    such that 
    $
        K \subseteq \bigcup_{i = 1}^N B(x_i, \tfrac{\delta}{2}).
    $
    In particular, for every $x \in K$ there exists $j \in \{1,\dots,N\}$ such that 
    $d(x, x_j) < \tfrac{\delta}{2}$.

    We claim that for all $i \in \N$ there exists $j \in \{1,\dots,N\}$ such that 
    $B(E_i, \tfrac{\delta}{2})$ contains $x_j$ and $\bigcup_{k = 0}^{i - 1}B(E_k, \tfrac{\delta}{2})$
    does not contain $x_j$.

    Indeed, since $E_i$ is non-empty, it contains a point $x$.
    There must exist a point $x_j$ satisfying $d(x, x_j) < \delta/2$.
    If $B(E_k, \tfrac{\delta}{2})$ with $k < i$ contains $x_j$, then 
    there exists $y \in E_k$ with $d(y, x_j) < \delta/2$.
    It follows that $d(y, x) < \delta$, contradicting the assumption that 
    $E_i$ is disjoint from $B(E_k, \delta)$.

    Hence, for all $j \in \{1, \dots, N\}$ there exists $i \in \{1,\dots,N\}$ such that $x_j \in B(E_i, \tfrac{\delta}{2})$.
    We obtain
    \[
        \bigcup_{i = 1}^N B(E_i, \delta)
        \supseteq 
        \bigcup_{i = 1}^N B(x_i, \tfrac{\delta}{2})
        \supseteq 
        K.
    \]
    Now, $E_{N + 1}$ is required to be a non-empty subset of $K$, disjoint from $K$. Contradiction!
\end{proof}

\begin{proposition}\label{Proposition: complement of a sequence of open sets}
    Let $A \subseteq \R^d$ be a closed set.
    Let $U_1, \dots, U_m \subseteq \R^d$ be open sets.
    Then the complement of the interior of 
    $A \setminus \bigcup_{i = 1}^m U_i$
    in $A$ is equal to
    $
        A \cap \bigcup_{i = 1}^m \overline{U}_i.
    $
\end{proposition}
\begin{proof}
    If a point is contained in $\bigcup_{i = 1}^m \overline{U}_i$, then it cannot be an interior point
    of $A \setminus \bigcup_{i = 1}^m U_i$, since every neighbourhood of that point will intersect some $U_i$.

    It remains to prove the converse inclusion.
    By contrapositive, it suffices to show that every point in $A$ which is not contained in 
    $
        \bigcup_{i = 1}^m \overline{U}_i
    $
    is an interior point of $A \setminus \bigcup_{i = 1}^m U_i$.

    If a point is not contained in $\bigcup_{i = 1}^m \overline{U}_i$, then its distance to this set is positive.
    Hence, a ball of sufficiently small radius will not intersect $\bigcup_{i = 1}^m \overline{U}_i$,
    and in particular it will not intersect any $U_i$, thus proving the claim.
\end{proof}

\begin{proof}[Proof of Lemma \ref{Lemma: main perturbation theorem}]
    Let $\varepsilon > 0$.
    Let $E_0 = \R^d \setminus \intr{K}$.
    Since $K$ is compact and $f$ is continuous, $f$ is uniformly continuous on $K$.
    There hence exists $\delta > 0$ such that 
    $
        \norm{x - y}_{\infty} \leq \delta \rightarrow \norm{f(x) - f(y)}_\infty < \tfrac{\varepsilon}{4}
    $ 
    for all $x, y \in K$.
    We may further assume that $\delta < \tfrac{\varepsilon}{4}$
    and that $x_0 \notin \overline{B}(E_0, \delta)$.

    Starting from $E_0$, we construct a finite sequence $E_0, \dots, E_N$ of non-empty disjoint closed subsets of $\R^d$ satisfying the following properties:
    \begin{enumerate}
        \item $x_0 \in \overline{B}(E_N, \delta)$.
        \item $x_0 \notin \overline{B}(E_i, \delta)$ for $i < N$.
        \item $E_i \subseteq K$ for all $i \geq 1$.
        \item $E_i$ is disjoint from $\bigcup_{j = 0}^{i - 1} B(E_j, \delta)$ for all $i$.
        \item If $N \geq i + 1$ and $x \in K \setminus \bigcup_{j = 0}^{i}B(E_j, \delta)$, then 
            $f(x) \in \overline{B}(E_{i}, \delta)$
            if and only if 
            $x \in E_{i + 1}$.
    \end{enumerate}

    Suppose we have constructed $E_0, \dots, E_m$ as above with $x_0 \notin \overline{B}(E_i, \delta)$ for $i = 0,\dots, m$.
    To construct $E_{m + 1}$, consider the closed set 
    $
        A = K \setminus \bigcup_{i = 0}^m B(E_i, \delta).
    $
    This set is non-empty, since it contains $x_0$.
    We claim that $x_0$ is contained in the interior of $A$.
    Indeed, suppose that $x_0$ is not contained in the interior of $A$.
    Then $B(x_0, \varepsilon)$ intersects $\bigcup_{i = 0}^m B(E_i, \delta)$ for all $\varepsilon > 0$.
    Since the union is finite, $B(x_0, \varepsilon)$ intersects some $B(E_i, \delta)$ for all $\varepsilon > 0$.
    It follows that $x_0$ is contained in the closure of $B(E_i, \delta)$.
    Since $E_i$ is closed,
    Proposition \ref{Proposition: closed balls about closed sets are closure of open balls} yields that the set $\overline{B}(E_i, \delta)$ is equal to
    the closure of the set $B(E_i, \delta)$.
    It follows that $x_0$ is contained in $\overline{B}(E_i, \delta)$, contradicting our assumptions about $E_0, \dots, E_m$.

    Since $A$ contains $x_0$ in its interior, and no robust invariant for $f$ does so, $A$ is not a robust invariant for $f$.
    We hence must have that the closed set
    \[
        E_{m + 1} = 
            \Set{ x \in K \setminus \bigcup_{i = 0}^m B(E_i, \delta) }
                { 
                    f(x) \notin A^\circ 
                }
            =
            \Set{ x \in K \setminus \bigcup_{i = 0}^m B(E_i, \delta) }
                { 
                    f(x) \in \bigcup_{i = 0}^m \overline{B}(E_i,\delta) 
                }
    \]
    is non-empty.
    The second equality above follows from Proposition \ref{Proposition: complement of a sequence of open sets}.

    By construction, $E_{m + 1}$ is disjoint from $\bigcup_{i = 0}^m B(E_i, \delta)$.
    Also by construction, if $x \in K \setminus \bigcup_{i = 0}^m B(E_i, \delta)$ then 
    $f(x) \in \bigcup_{i = 0}^m \overline{B}(E_i,\delta)$
    if and only if $x \in E_{m + 1}$.
    We claim that for $x \in K \setminus \bigcup_{i = 0}^m B(E_i, \delta)$ we have 
    $f(x) \in \bigcup_{i = 0}^m \overline{B}(E_i,\delta)$ if and only if 
    $f(x) \in E_{m}$.
    Indeed we have for all $k < m$ that $x \in K \setminus \bigcup_{i = 0}^k B(E_i, \delta)$,
    so that $f(x) \in B(E_k, \delta)$ implies $x \in E_{k + 1}$, which is impossible by our assumption on $x$.
    This proves the claim.
    We have now shown that $E_{m + 1}$ has all the required properties.

    If $x_0 \in \overline{B}(E_{m + 1}, \delta)$, then we are done.
    Otherwise we construct the next element of the sequence using the same method.

    By 
        Lemma \ref{Lemma: disjoint sequence of sets in a compact set},
    the sequence $E_0, E_1, \dots$ is finite, \textit{i.e.}, there must exist $N \in \N$ such that $x_0 \in \overline{B}(E_N, \delta)$.

    We now construct an $\varepsilon$-perturbation $\widetilde{f}$ of $f$ under which the point $x_0$ escapes $K$.
    We will construct a sequence of points $x_1, \dots, x_{N}$ and continuous functions $f_1, f_2, \dots, f_{N}$ such that:
    \begin{enumerate}
        \item $f_i$ agrees with $f$ outside $\bigcup_{j = 0}^{i - 1} B(x_j, \tfrac{\delta}{2})$ for all  $i \in \{1, \dots, N\}$.
        \item $f_i$ agrees with $f_{i - 1}$ on $\bigcup_{j = 0}^{i - 2} B(x_j, \tfrac{\delta}{2})$  for all  $i \in \{2, \dots, N\}$.
        \item $\norm{f_i - f}_\infty < \varepsilon$ for all $i \in \{1, \dots, N\}$.
        \item $f_i(x_{i - 1}) = x_i$ for all $i \in \{1, \dots, N\}$.
        \item $x_i \in E_{N - i}$ for all $i \in \{1, \dots, N\}$.
        \item $x_N \notin K$.
    \end{enumerate}
    Observe that the balls $B(x_j, \tfrac{\delta}{2})$ are disjoint for $0 \leq j < N$.
    Indeed,
    if $j > k > 0$ we have $x_j \in E_{N - j}$ and $x_k \in E_{N - k}$, so that 
    $d(x_j, x_k) \geq \delta$ by the third property of the $E_i$s.
    Further, we have $x_0 \notin \overline{B}(E_{N - j}, \delta)$ for all $j > 0$ by the second property of the $E_i$s.
    Hence, $d(x_0, x_j) > \delta$ for all $j > 0$.

    It then follows that $\widetilde{f} := f_{N + 1}$ is an $\varepsilon$-perturbation of $f$ with 
    $
        \widetilde{f}(x_i) = f_{N + 1}(x_i) = f_{i + 1}(x_i) = x_{i + 1}
    $
    for all $i = 0, \dots, N$.
    Since $x_{N} \notin K$, the point $x_0$ escapes $K$ under $\widetilde{f}$.

    It remains to construct the sequences as above.
    Let us first construct $x_1$ and $f_1$.
    Since $x_0 \in \overline{B}(E_N, \delta)$, there exists $\widetilde{x}_0 \in E_N$ with 
    $d(x_0, \widetilde{x}_0) \leq \delta$.
    We have 
        $f\left(\widetilde{x}_0\right) \in B(E_{N - 1}, \delta)$, 
    so that there exists a point $x_1 \in E_{N - 1}$ with 
    $d(f(\widetilde{x}_0), x_1) < \delta$.
    If $N = 1$, we have $E_{N - 1} = \R^d \setminus \intr{K}$, so that we may also assume that $x_1 \notin K$.

    Define a new function $f_1 \colon K \to \R^d$ as follows:
    \[
        f_1(x) = 
            \begin{cases}
                x_1 \left(1 - \frac{2d(x, x_0)}{\delta}\right) + f(x) \frac{2d(x,x_0)}{\delta}
                & \text{if } d(x, x_0) < \delta/2,\\
                f(x) &\text{otherwise.}
            \end{cases}
    \]
    The function $f_1$ is continuous, maps $x_0$ to $x_1$, and agrees with $f$ outside $B(x_0, \delta/2)$.
    On $B(x_0, \delta/2)$ we have:
    \begin{align*}
            \left|f_1(x) - f(x)\right|
            &=
            \left(1 - \frac{2d(x, x_0)}{\delta}\right)
            \norm{
                x_1 - f(x)
            }\\
            &\leq 
            \norm{
                x_1 - f(\widetilde{x}_0)
            }
            +
            \norm{
                f(\widetilde{x}_0) - f(x_0)
            }
            +
            \norm{
                f(x_0) - f(x)
            }\\
            &< \delta + \varepsilon/4 + \varepsilon/4
            < 3 \varepsilon/4
            < \varepsilon.
    \end{align*}
    Hence, $f_1$ is an $\varepsilon$-perturbation of $f$.

    Suppose we have constructed $f_1, \dots, f_i$ and $x_1,\dots,x_{i}$ as above.
    We have $f_i(x_i) = f(x_i) \in \overline{B}(E_{N - i - 1}, \delta)$ by the fourth property of the $E_i$s, 
    since $x_i$ is contained in $E_{N - i}$ by assumption.
    There hence exists $x_{i + 1} \in E_{N - i - 1}$ with 
    $d(f_i(x_i), x_{i + 1}) \leq \delta$.
    If $i + 1 = N$, then we may also assume that $x_{i + 1} \notin K$,
    at the cost of correcting the estimate to $d(f_i(x_i), x_{i + 1}) \leq 2\delta$.

    Define 
    \[
        f_{i + 1}(x) =
        \begin{cases}
            x_{i + 1} \left(1 - \frac{2d(x, x_i)}{\delta}\right) + f(x) \frac{2d(x,x_i)}{\delta}
            & \text{if } d(x, x_i) < \delta/2,\\
            f_i(x) &\text{otherwise.}
        \end{cases}
    \]
    It is clear that $f_{i + 1}$ maps $x_i$ to $x_{i + 1}$ and that it agrees with with $f_i$ outside $B(x_i, \tfrac{\delta}{2})$.
    The map is continuous since $f_i$ agrees with $f$ on $B(x_i, \tfrac{\delta}{2})$.
    By induction hypothesis, we have $d(f_{i + 1}(x), f(x)) < \varepsilon$ for all $x \notin B(x_i, \tfrac{\delta}{2})$.
    For $x \in B\left(x_i, \tfrac{\delta}{2}\right)$ we have:
    \begin{align*}
        \norm{
            f_{i + 1}(x) - f(x)
        }
        &=
        \left|1 - \frac{2d(x, x_i)}{\delta}\right|
        \norm{
            x_{i + 1} - f(x)
        }\\
        &\leq 
        \norm{x_{i + 1} - f(x_i)}_\infty
        + \norm{f(x_i) - f(x)}_\infty
        \leq 
        2\delta 
        + \varepsilon/4
        <
        \varepsilon.
    \end{align*}
    This shows that $f_{i + 1}$ has all the required properties.
\end{proof}

\section{Proof of Lemma \ref{Lemma: generic termination of point-escape algorithm}}\label{Appendix: proof of generic termination of point-escape algorithm}

As a preparation we observe:

\begin{proposition}\label{Proposition: pulling back robust invariant}
    Let $A$ be a closed set. 
    Let $x_0 \in \intr{A}$.
    Assume that there exist
    a robust invariant $V$ for $f$ and a number
    $n \in \N$ such that $f^i(x_0) \in \intr{A}$ for all $1 \leq i \leq n - 1$, and 
    $f^n(x_0) \in \intr{V} \subseteq V \subseteq \intr{A}$.
    Then there exists a robust invariant $V_0$ for $f$ with 
    $x_0 \in \intr{V_0} \subseteq V_0 \subseteq \intr{A}$.
\end{proposition}
\begin{proof}
    By induction, it suffices to show the result for $n = 1$.
    Assume that $f(x_0)$ is contained in the interior of a robust invariant for $f$.
    Then $f(x_0) \in \intr{V}$, so that there exists $\varepsilon > 0$ such that 
    $B(f(x_0), \varepsilon) \subseteq \intr{V}$.
    By continuity, there exists 
    $\delta > 0$ such that 
    \[
        f\left(\overline{B}(x_0,\delta)\right) \subseteq B(f(x_0), \varepsilon)\subseteq \intr{V}.
    \]
    Since $x_0 \in \intr{A}$, we may assume that $\delta$ is so small that $\overline{B}(x_0, \delta) \subseteq \intr{A}$.
    Then $V_0 = \overline{B}(x,\delta) \cup V$ has all the required properties.
\end{proof}

\begin{proposition}\label{Proposition: invariant for n-fold iterate is contained in invariant for f}
    Let $f \colon \R^d \to \R^d$. 
    Let $V \subseteq \R^d$ be a compact subset.
    If $V$ is a robust invariant for $f^N(V)$ for some $N \geq 1$, then 
    there exists $V_0 \supseteq V$ with 
    $f(V_0) \subseteq \intr{V}_0$.

    Moreover, if $U$ is an open set such that $f^k(V) \subseteq U$ for all $k \leq N$,
    then $V_0$ can be chosen such that $V_0 \subseteq U$. 
\end{proposition}
\begin{proof}
    The sets $f^k(V)$ are compact for all $k \in \N$.
    There hence exists a positive integer $B$ with $f^k(V) \subseteq [-B,B]^d$
    for all $k \leq N$.
    Then $f$ is uniformly continuous on $[-B - 1,B + 1]^d$.
    There hence exists a function
    $
        \omega_f \colon \N \to \N,
    $
    satisfying 
    \[
        d(x,y) \leq 2^{-\omega_f(n)}
        \rightarrow
        d(f(x),f(y)) < 2^{-n}
    \]
    for all $x,y \in [-B - 1, B + 1]^d$.

    Since $f^N(V)$ is compact and $\intr{V}$ is open, there exists $n \in \N$ such that 
    $B(f^N(V), 2^{-n}) \subseteq \intr{V}$.
    Consider the set 
    \[
        V_0 = V \cup \bigcup_{k = 1}^{N - 1} \overline{B}\left(f^{k}(V), 2^{-\omega_f^{N - k}(n)}\right).
    \]
    For $k < N - 1$ we have 
    \[
        f\left(\overline{B}\left(f^{k}(V), 2^{-\omega_f^{N - k}(n)}\right)\right)
        \subseteq 
        B\left(f^{k + 1}(V), 2^{-\omega_f^{N - k - 1}(n)}\right)
        \subseteq 
        \intr{V}_0.
    \]
    Indeed, if $x \in \overline{B}\left(f^{k}(V), 2^{-\omega_f^{N - k}(n)}\right)$,
    then $x \in [-B-1, B+1]^d$.
    Pick $y \in f^{k}(V)$ with $d(x,y) < 2^{-\omega_f^{N - k}(n)}$.
    Then $f(y) \in f^{k + 1}$ with 
    $d(f(x), f(y)) < 2^{-\omega_f^{N - k - 1}(n)}$.

    For $k = N - 1$ we have 
    \[
        f\left(\overline{B}\left(f^{N - 1}(V), 2^{-\omega_f(n)}\right)\right)
        \subseteq 
        B(f^{N}(V), 2^{-n})
        \subseteq 
        \intr{V}
        \subseteq
        \intr{V}_0.
    \]

    If additionally $f^k(V) \subseteq U$ for some open set $U$ and all $0 \leq k \leq N$,
    then by compactness of the sets $f^k(V)$, we can ensure that $V_0 \subseteq U$
    up to choosing $\omega_f$ such that the the balls $\overline{B}\left(f^{k}(V), 2^{-\omega_f^{N - k}(n)}\right)$
    are all contained in $U$. 
    This is possible since we may always replace $\omega_f$ with a pointwise larger function.
\end{proof}

We now prove Lemma \ref{Lemma: generic termination of point-escape algorithm}. 

\begin{proof}[Proof of Lemma \ref{Lemma: generic termination of point-escape algorithm}]
    Since Algorithm \ref{Algorithm: main algorithm} halts on all escaping instances, the instance $(f, K, x_0)$ must be trapped.

    Let $\varepsilon > 0$.
    Consider the compact set $K_{\varepsilon} = \bigcup_{y \in K}\overline{B}(y,\varepsilon)$.
    Define the sequence $(x_i)_i$ recursively via $x_{i + 1} = f(x_i)$.
    Then $(x_i)_i$ is contained in the interior of $K_{\varepsilon}$.
    We construct an $\varepsilon$-perturbation $f_{\varepsilon}$ of $f$ such that $x_0$ is robustly trapped in $K_{\varepsilon}$ under $f_{\varepsilon}$.

    Since $K_{\varepsilon} \to K$ and $f_{\varepsilon} \to f$ as $\varepsilon \to 0$, the claim follows.
    
    Now, assume that $f(x_i) \in \intr{K}$ for all $i \in \N$.
    Since $K_{\varepsilon}$ is compact,
    there exist $n < m \in \N$ such that $\norm{x_m - x_n}_\infty < \varepsilon$.
    
    We may assume that $\norm{x_i - x_j}_\infty \geq \varepsilon$ for all 
    $i,j \in \{n, \dots, m\}$ with $i < j$ and $i > n$ or $j < m$.

    Choose a positive number $\delta > 0$ such that the balls 
    $\overline{B}(x_i, \delta)$
    are disjoint and contained in $\intr{K_\varepsilon}$
    for $i = n, \dots, m - 1$ and such that we have
    $
        d(x,y) < \delta \rightarrow d(f(x), f(y)) < \varepsilon
    $
    for all $x, y \in K$.

    Let $f_\varepsilon$ be the function which is obtained from $f$ as follows:
    For $i = n, \dots, m - 2$, replace $f$ on $B(x_i, \delta)$ with the function
    \[
        \widetilde{f}_i(x) =  
        \begin{cases}  
            x_{i + 1} &\text{if }\norm{x - x_{i}}_\infty < \tfrac{\delta}{2}\\
            \left(1 - \tfrac{2 d(x, x_i) - \delta}{\delta}\right) x_{i + 1} 
            + \tfrac{2 d(x, x_i) - \delta}{\delta}f(x)  &\text{if }\norm{x - x_{i}}_\infty \geq \tfrac{\delta}{2}\\
        \end{cases}
    \]
    On $B(x_{m - 1}, \delta)$, replace $f$ with the function
    \[
        \widetilde{f}_{m - 1}(x) = 
            \begin{cases}
                x_n &\text{if }\norm{x - x_{m - 1}}_\infty < \tfrac{\delta}{2}\\
                \left(1 - \tfrac{2 d(x, x_{m - 1}) - \delta}{\delta}\right) x_{n} 
            + \tfrac{2 d(x, x_{m - 1}) - \delta}{\delta}f(x) &\text{if }\norm{x - x_{m - 1}}_\infty \geq \tfrac{\delta}{2}\\
            \end{cases}
    \]
    Let $i \in \{n, \dots, m - 2\}$.
    Then, for $x \in B(x_i, \delta/2)$, we have: 
    \[
            \left|
                f(x) - f_\varepsilon(x)    
            \right|
            =
            \left|
                f(x) - x_{i + 1}    
            \right|
            =
            \left|
                f(x) - f(x_i) 
            \right|
            <
            \varepsilon.
    \] 
    For $x \in B(x_i,  \delta) \setminus B(x_i,  \delta/2)$ we have:
    \[
            \left|
                f(x) - f_{\varepsilon}(x)    
            \right|
            =
            \left|1 - \tfrac{2 d(x, x_i) - \delta}{\delta}\right|
            \left|
                f(x) - f(x_{i})   
            \right|
            <
            \varepsilon.
    \] 
    The case of $i = m - 1$ is identical.

    We have 
    \[
        f_{\varepsilon}
        \left(
            \bigcup_{i = n}^{m - 1} \overline{B}\left( x_i, \tfrac{\delta_i}{2} \right)
        \right)
        = 
        \{
            x_n, \dots, x_{m - 1}
        \}
        \subseteq 
        \bigcup_{i = n}^{m - 1} B\left( x_i, \tfrac{\delta_i}{2} \right).
    \]
    Hence, $V = \bigcup_{i = n}^{m - 1} \overline{B}\left( x_i, \tfrac{\delta_i}{2} \right)$ is a robust invariant for $f_{\varepsilon}$.
    By construction, $V \subseteq \intr{K}$. 
    Further, $V$ contains $f_{\varepsilon}^n(x_0)$ in its interior, so that by Proposition \ref{Proposition: pulling back robust invariant}, 
    $f(x_0)$ is contained in the interior of a robust invariant for $f_{\varepsilon}$.
    It follows from Theorem \ref{Theorem: main theorem} that $(f_{\varepsilon}, K_{\varepsilon}, x_0)$ is a robust trapped instance of the Point Escape Problem.
\end{proof}

\section{Proof of Lemma \ref{Lemma: classification of robust instances}}
\label{Appendix: Proof of classification of robust instances}

We first show that the instances satisfying the conditions listed in Lemma \ref{Lemma: classification of robust instances} are robust.
Assume that $r(A) < 1$. 
Then $A$ does not admit $1$ as an eigenvalue, so that the map $(I - A)$ is invertible.
It follows that $Ax + b$ has a unique fixed point.
A calculation shows that this fixed point is given by $(I - A)^{-1}b$.
We have for all $x \in \R^d$:
\[
    f_{A,b}^n(x) = A^nx + \sum_{i = 0}^{n - 1}A^i b 
                = A^nx + (I - A^n) (I - A)^{-1} b
                \to (I - A)^{-1}b.
\]
On a sufficiently small neighbourhood of $A$, the point $(I - A)^{-1}b$ depends continuously on $A$.
It is a routine exercise to obtain a computable bound on $n$ in terms of $A$, $b$, $x_0$, and $\varepsilon > 0$ such that 
$\norm{f_{A,b}^n(x_0) - (I - A)^{-1}b} < \varepsilon$.
Now, if the point 
$(I - A)^{-1}b$ 
and all points in the orbit of $x_0$ under $f_{A,b}$
are contained in the interior of $P$, 
we obtain that the point must be robustly trapped.

Now assume that the orbit $(f_{A,b}^n(x_0))_n$ is contained in the interior of $P$, 
that $A$ has a simple real eigenvalue $\rho > 1$ with $\rho > |\lambda|$ for all other eigenvalues and 
that there exists an eigenvector $v$ for $\rho$ with $v \in \intr{P}_0$, $x_0 = \alpha v + w$ where $\alpha > 0$ and $w$ is a linear combination of eigenvalues and generalised eigenvalues for eigenvalues $\lambda \neq \rho$.

We will work with a homogenized version of the system. Consider the matrix 
\[
    B = \begin{bmatrix}
        A & b \\ 
        0 & 1
    \end{bmatrix} \in \R^{(d + 1) \times (d + 1)}.
\]
Consider 
the initial point $z_0 = (x_0, 1) \in \R^{d + 1}$,
and the polyhedron 
$
    Q = \bigcap_{i = 1}^m J_i
$
where $J_i = \Set{z \in \R^{d + 1}}{(N_i, -D_i) \cdot z \leq 0}$.
We have $\sigma(B) = \sigma(A) \cup \{1\}$, so that $\rho$ satisfies 
$\rho > |\lambda|$ for all other eigenvalues of $B$.

Write $z_0 = \alpha (v,0) + (w,1)$.
Consider the sequence 
$
    \tfrac{1}{\rho^n}B^n z_0 = \alpha (v,0) + \tfrac{1}{\rho^n} B^n w.
$
Then, by a standard argument involving the Jordan normal form (cf.~\cite{Tiwari}), the entries of 
$B^n w$ are bounded by an exponential polynomial of the form 
$\sum_j \sum_k t_{j,k} n^k |\lambda_j|^n$
where $|\lambda_j| < \rho$ for all $j$, so that we have 
$
    \tfrac{1}{\rho^n}B^n z_0 \to \alpha (v,0).
$
It follows that 
$
    \tfrac{1}{\rho^n}f^n_{A,b}(x_0) \to \alpha v.
$
Now, we have $N_i f^n_{A,b}(x_0) \leq D_i$ if and only if 
$\tfrac{1}{\rho^n} N_i f^n_{A,b}(x_0) \leq \tfrac{D}{\rho^n}$.
Since $N_i v < 0$ for all $i$, this inequality holds true for all sufficiently large $i$.
Again, it is a routine exercise to derive effective bounds on $n$ in terms of the data that 
show that the instance is robustly trapped.
See \cite[Lemma 6.1]{DPLRSR21} for a similar result.

Finally, assume that $A$ has the simple real eigenvalue $\rho = 1$ with $\rho > |\lambda|$ for all other eigenvalues,
there exists an eigenvector $v$ for $\rho$ with $v \in \intr{P}_0$ such that 
$x_0 = \alpha v + w$,
$b = \beta v + u$,
where $\alpha, \beta > 0$,
and $w$ and $u$ are linear combinations of eigenvalues and generalised eigenvalues
for eigenvalues $\lambda \neq \rho$.

A similar argument as above shows that we have $\tfrac{1}{n} f^n_{A,b}(x_0) \to \beta v$ as $n \to \infty$,
so that the point $x_0$ is trapped under $f_{A,b}$.
Under small perturbations of $A$, the largest real eigenvalue $\rho$ can be perturbed to an eigenvalue $\rho < 1$ or an eigenvalue $\rho > 1$,
but if the perturbations are sufficiently small, we will still have the relation $\rho > |\lambda|$ for all other eigenvalues $\lambda$.
For perturbations $\widetilde{A}$ with $\rho > 1$, the sequence $f^n_{\widetilde{A},b}(x_0)$ asymptotically behaves like $\alpha v$,
so that the point $x_0$ remains trapped.
For perturbations $\widetilde{A}$ with $\rho < 1$, the sequence $f^n_{\widetilde{A},b}(x_0)$ converges to the fixed point 
$(I - \widetilde{A})^{-1}b$.
Using $b = \beta v + u$, we obtain 
$(I - \widetilde{A})^{-1} b = \frac{\beta}{1-\rho}v + (I - \widetilde{B})^{-1} u$, where all eigenvalues of the matrix $\widetilde{B}$ have a strictly smaller absolute value than $\rho$.
Hence, 
$\frac{(I - \widetilde{A})^{-1} b}{\norm{(I - \widetilde{A})^{-1} b}} \to \beta v$
as $\rho \to 1$, so that the sequence remains trapped under all sufficiently small perturbations of this type.
Here, we can derive effective bounds on $n$ in terms of  $\varepsilon > 0$ and the data such that, no matter which of the above is the case, the point 
$\tfrac{f^n_{\widetilde{A},b}(x_0)}{\norm{f^n_{\widetilde{A},b}(x_0)}}$ 
will be $\varepsilon$-close to a positive multiple of $v$.
This yields that the instance is robustly trapped.

Conversely, assume that the instance is robustly trapped.
Then the orbit $(f^n_{A,b}(x_0))_n$ must be contained in the interior of $P$, for otherwise the instance 
escapes under arbitrarily small perturbations of $P$.

If $r(A) < 1$, then as we saw above, the map $f_{A,b}$ has the unique fixed point 
$(I - A)^{-1}b$
with $f^n_{A,b}(x_0) \to (I - A)^{-1}b$.
Since the instance is robustly trapped we must have $(I - A)^{-1}b \in \intr{P}$.

Consider the case $r(A) \geq 1$.
Assume that the instance does not satisfy any of the assumptions of Lemma \ref{Lemma: classification of robust instances}. 
We distinguish four cases:
\begin{enumerate}
    \item There is no real eigenvalue $\geq 1$.
    \item The largest real eigenvalue is $\geq 1$ but is not simple
    or not strictly greater than all other eigenvalues.
    \item The largest real eigenvalue is simple, greater than all other eigenvalues, and equal to $1$,
    but all eigenvectors $v$ for $\rho$ satisfy $v \notin \intr{P}_0$ or 
    in the decompositions $x_0 = \alpha v + w$ and $b = \beta v + u$ as above we have
    $\alpha \leq 0$ or $\beta \leq 0$.
    \item The largest real eigenvalue $\rho$ is simple, greater than all other eigenvalues, strictly greater than $1$,
        but all eigenvectors $v$ for $\rho$ satisfy $v \notin \intr{P}_0$ or in the decomposition $x_0 = \alpha v + w$ as above we have
        $\alpha \leq 0$.
\end{enumerate}
We claim that the instance is not robustly trapped in any of the above cases.

In the last case, it easily follows from our argument above that the point $x_0$ eventually escapes under $f_{A,b}$,
potentially up to an arbitrarily small perturbation of the eigenspace of $\rho$.
Indeed, fix an eigenvector $v$ for $\rho$ of Euclidean norm $1$.
Up to an arbitrarily small perturbation we may assume that  $\alpha \neq 0$.
Then $\tfrac{1}{\rho^n}f^n_{A,b}(x_0) \to \alpha v$ where for at least one normal $N_i$ of $P$ we have 
$N_i \alpha v \geq 0$.
Thus, up to a small perturbation of $v$ we can ensure $N_i \alpha v > 0$, making the instance escape.

The second-to-last case is similar.
Fix an eigenvector $v$ for $\rho$ of Euclidean norm $1$.
Up to an arbitrarily small perturbation we may assume that $\alpha \neq 0$ and that  $\beta \neq 0$.
We have, potentially up to an arbitrarily small perturbation, $N_i \alpha v > 0$ or $N_i \beta v > 0$.
Again, potentially up to a small perturbation, we can ensure that 
$\tfrac{1}{\rho^n}f^n_{A,b}(x_0) \to \alpha v$
or that 
$\tfrac{1}{\rho^n}f^n_{A,b}(x_0) \to \beta v$
(where $\rho$ is the largest real eigenvalue of the perturbed instance),
so that the point escapes.

In the two remaining cases, under arbitrarily small perturbations of the matrix, the largest positive real eigenvalue $\rho$ satisfies 
$\rho < |\lambda|$ for some complex or negative eigenvalue $\lambda$ with $|\lambda| > 1$.
Up to an arbitrarily small perturbation, $x_0$ has a non-zero component in an eigenspace of $\lambda$.
It then follows that $f_{A,b}^n(x_0)$ escapes every affine half-space \cite{Tiwari,BellGerhold07}. 
The results in \cite{Tiwari,BellGerhold07} are stated for homogeneous linear systems, but we can homogenize the system as in the construction of $B$ above, introducing the real eigenvalue $1$. 
Since $|\lambda| > 1$, the result still applies.

\end{document}

%% file: commands.tex
\newcommand{\norm}[1]{\left|\left|#1\right|\right|}

\newcommand{\dom}{\operatorname{dom}}
\newcommand{\id}{\operatorname{id}}
\newcommand{\R}{\mathbb{R}}
\newcommand{\C}{\mathbb{C}}
\newcommand{\Q}{\mathbb{Q}}
\newcommand{\N}{\mathbb{N}}
\newcommand{\Z}{\mathbb{Z}}
\newcommand{\NN}{\N^{\N}}

\newcommand{\diam}{\operatorname{diam}}
\newcommand{\QD}{\mathscr{Q}}

\renewcommand{\O}{\mathcal{O}}

\newcommand{\subf}{\subseteq_{\text{\tiny fin}}}

\newcommand{\intr}[1]{#1^{\circ}}

\newlength\arrowheight

\newcommand{\Set}[2]{\left\{#1 \mid #2\right\}}



%% file: Linear-Systems.tex
\section{Linear Systems}\label{Section: Linear Systems}

As mentioned in the introduction, Algorithm \ref{Algorithm: main algorithm} may fail to terminate on very simple-looking problem instances when $A$ is not a compact set.
The main reason for this is that our representation of continuous functions induces the topology of uniform convergence on compact sets, or equivalently, the topology induced by the metric \eqref{eq: distance on functions}.
In order for two functions $f, g \colon \R^d \to \R^d$ to have distance at most $2^{-n}$, it suffices that $f$ and $g$ agree on the cube $[-2^{n+1},2^{n+1}]^d$,
potentially admitting function values of arbitrarily large distance outside of this cube.
Thus, every function $f$ has ``arbitrarily small'' perturbations that differ from $f$ arbitrarily outside a certain bounded set.

Revisiting the example from the introduction, consider the set $A = \Set{x \in \R}{x \geq 0}$, the initial point $x_0 = 1 \in A$, and the map 
$f \colon \R \to \R$,  $f(x) = 2x$.
It is clear that the point $x_0$ is trapped in $A$ under iteration of $f$. 
However, it cannot be robustly trapped, since $f$ does not admit any robust invariants.
More explicitly, with respect to the metric \eqref{eq: distance on functions}, $f$ is the limit of the sequence
\[
    f_n(x) = \min\left\{2x, \left(1 - 2^{n + 1}\right)x + \left(1 + 2^{n + 1}\right)\left(2^n - 1\right)\right\}.
\]
We have $f_n^{n}(1) = 2^n$ and $f_n(2^n) = -1$, so that $x_0$ escapes $A$ under each $f_n$.

While a na\"ive direct application of Algorithm \ref{Algorithm: main algorithm} fails to verify that the system is trapped,
we can exploit the fact that the particular $f$ and $A$ given above are well behaved ``near infinity'' to compute a compactified version of the system that is amenable to analysis via Algorithm \ref{Algorithm: main algorithm}.

Consider the map
$
    \Phi \colon \R \to (-1,1),\;
    \Phi(x) = \tfrac{x}{1 + |x|}
$.
This map has a continuous inverse, explicitly given by
$
    \Phi^{-1} \colon (-1,1) \to \R,\;
    \Phi^{-1}(x) = \tfrac{x}{1 - |x|}
$.
We have
$
    \Phi \circ f \circ \Phi^{-1}(x) = \tfrac{2x}{1+|x|}
$.
The expression on the right-hand side of this equality defines a total function $\overline{f}(x) \colon \R\to\R$.
Observe that a point $x$ is trapped in $A$ under $f$ if and only if $\Phi(x)$ is trapped in $A$ under $\overline{f}$.
The latter can be verified by Algorithm \ref{Algorithm: main algorithm} for all $x > 0$, 
since $\overline{f}$ admits the robust invariant $\left[\tfrac{\Phi(x)}{2},2\right]$
with $\Phi(x) \in \left(\tfrac{\Phi(x)}{2},2\right) \subseteq \left[\tfrac{\Phi(x)}{2},2\right] \subseteq \intr{A}$.

To further illustrate how the scope of Algorithm \ref{Algorithm: main algorithm} can be extended using this compactification technique, 
we generalise the above observation to give a solution for the problem of deciding whether a point escapes a polyhedron (which is in general unbounded) under the iteration of an affine linear map:

\begin{definition}
    The \emph{Linear Escape Problem} asks to determine for a given non-singular matrix $A \in \R^{d \times d}$,
    a given vector $b \in \R^d$, a given polyhedron $P \subseteq \R^d$, and a given point $x_0 \in \R^d$
    whether $x_0$ escapes $P$ under the map $f_{A,b}(x) = Ax + b$.

    Here, $A$ is given by a $\rho^{d \times d}$-name of its entries,
    and $P$ is given by a finite list of affine half-spaces $H_1, \dots, H_m$
    with $P = \bigcap_{i = 1}^m H_i$,
    where each $H_i$ is given by a $\rho^d$-name of a normal $0 \neq N_i \in \R^d$ and a $\rho^1$-name of a distance $D_i \in \R$,
    such that $H_i = \Set{x \in \R^d}{N_ix \leq D_i}$.
\end{definition}

We will give a reduction of the Linear Escape Problem to the Point Escape Problem that sends robust instances to robust instances
based on the compactification idea discussed above.

Before we describe the reduction, we classify the robust instances of the Linear Escape Problem.
Since all escaping instances are robust, we focus only on trapped instances.

For a matrix $A$, let $\sigma(A) \subf \C$ denote the set of its eigenvalues and $r(A) = \max \Set{|\lambda|}{\lambda \in \sigma(A)}$ denote the largest value among the absolute values of its eigenvalues.

For a polyhedron $P = \bigcap_{i = 1}^m H_i$
with 
$H_i = \Set{x \in \R^d}{N_i x \leq D}$, 
let 
$P_0 = \bigcap_{i = 1}^m H_i^0$, 
where 
$H_i^0 = \Set{x \in \R^d}{N_i x \leq 0}$.

\begin{lemma}\label{Lemma: classification of robust instances}
    Let $(A, b, P, x_0)$ be a trapped instance of the Linear Escape Problem.
    Then $(A, b, P, x_0)$ is robust if and only if the orbit $\left(f_{A,b}^n(x_0)\right)_n$ is contained in the interior of $P$
    and one of the three following conditions is met:
    \begin{enumerate}
        \item $r(A) < 1$ and the unique fixed point of the map $f_{A,b}(x) = Ax + b$ is contained in $\intr{P}$.
        \item $A$ has a simple real eigenvalue $\rho > 1$ with $\rho > |\lambda|$ for all other eigenvalues and 
        there exists an eigenvector $v$ for $\rho$ with $v \in \intr{P}_0$ and 
        $x_0 = \alpha v + w$ where $\alpha > 0$ and $w$ is a linear combination of eigenvalues and generalised eigenvalues
        for eigenvalues $\lambda \neq \rho$ (in the sense of the real Jordan normal form).
        \item $A$ has the simple real eigenvalue $\rho = 1$ with $\rho > |\lambda|$ for all other eigenvalues,
        there exists an eigenvector $v$ for $\rho$ with $v \in \intr{P}_0$ such that 
        $x_0 = \alpha v + w$,
        $b = \beta v + u$,
        where $\alpha, \beta > 0$,
        and $w$ and $u$ are linear combinations of eigenvalues and generalised eigenvalues
        for eigenvalues $\lambda \neq \rho$ (in the sense of the real Jordan normal form).
    \end{enumerate}
\end{lemma}
\begin{proof}
    The proof is very similar to that of \cite[Proposition 5.1, Proposition 6.2.]{DPLRSR21}.
    It is given in Appendix \ref{Appendix: Proof of classification of robust instances}.
\end{proof}

We now describe the reduction to the Point Escape Problem.
Consider the map 
\[
    \Phi \colon \R^d \to B(0,1),
    \;
    \Phi(x) = \frac{x}{1 + \norm{x}}.
\]
This map is invertible with inverse 
\[
    \Phi^{-1} \colon B(0,1) \to \R^d,
    \;
    \Phi^{-1}(x) = \frac{x}{1 - \norm{x}}.
\]
The map $\Phi \circ f_{A,b} \circ \Phi^{-1} \colon B(0,1) \to B(0,1)$ 
extends (uniformly computably) to the continuous map 
\[
    \overline{f}_{A,b} \colon \R^d \to \R^d,
    \;
    \overline{f}_{A,b}(x) 
        = 
    \frac{Ax + \left(1 - \min\left\{1,\norm{x}\right\}\right)b}{\norm{Ax + \left(1 - \min\left\{1,\norm{x}\right\}\right)b} + 1 - \min\left\{1, \norm{x}\right\}}.
\]

For a polyhedron $P = \bigcap_{i = 1}^m H_i$ with $H_i = \Set{x \in \R^d}{Nx \leq D}$, define 
$\widehat{P} = \bigcap_{i = 1}^m \widehat{H}_i$, where 
\[\widehat{H_i} = \Set{x \in \R^d}{Nx \leq D \left(1 - \min\{1, \norm{x}\}\right)}.\]
Observe that $\Phi(P) = \widehat{P} \cap B(0, 1)$.

\begin{proposition}
    The point $x_0$ is trapped in $P$ under $f_{A,b}$ if and only if 
    the point $\frac{x_0}{1 + \norm{x_0}}$ is trapped in $\widehat{P}$ under $\overline{f}_{A,b}$.
\end{proposition}
\begin{proof}
    Since $\Phi$ is bijective we have $x \in P$ if and only if $\Phi(x) \in \Phi(P) = \widehat{P} \cap B(0,1)$.
    On $B(0,1)$, the map $\overline{f}_{A,b}$ agrees with 
    $\Phi \circ f_{A,b} \circ \Phi^{-1}$,
    which maps $B(0,1)$ to $B(0,1)$.
    It follows that we have
    $\overline{f}_{A,b}^n\left(\Phi\left(x_0\right)\right) = \Phi \circ f^n_{A,b} (x_0)$.
    This yields the claim.
\end{proof}

We now obtain our desired reduction:

\begin{proposition}
    The instance $(A, b, P, x_0)$ is a robust instance of the Linear Point Escape Problem 
    if and only if the instance $(\overline{f}_{A,b}, \widehat{P}, \Phi(x_0))$ is a robust 
    instance of the Point Escape Problem.
\end{proposition}
\begin{proof}
    If the instance escapes, the claim is obvious.
    Thus, assume that the instance is trapped.
    We show that $\overline{f}_{A,b}$ has a robust invariant in the interior of 
    $\widehat{P}$.

    Consider the classification of robust instances given in Lemma \ref{Lemma: classification of robust instances}.

    In the first case, it follows from the proof of Lemma \ref{Lemma: classification of robust instances}
    that the point
    $p = \Phi((I - A)^{-1}b)$ is a fixed point of 
    $\overline{f}_{A,b}$
    such that every point in $B(0,1)$ converges to $\Phi((I - A)^{-1}b)$
    under the iteration of $\overline{f}_{A,b}$.
    Moreover, for a sufficiently small $\varepsilon > 0$ and 
    sufficiently large $N$, the ball 
    $\overline{B}(p, \varepsilon)$
    is mapped by 
    $\overline{f}_{A,b}^N$
    inside the ball 
    $B(p, \varepsilon/2)$.

    In the second and third case, it follows from the proof of Lemma \ref{Lemma: classification of robust instances}
    that the point $p = \tfrac{v}{\norm{v}}$ is a fixed point of $\overline{f}_{A,b}$
    with $\overline{f}_{A,b}^n\left(\Phi(x_0)\right) \to p$ as $n \to \infty$.
    Moreover, for a sufficiently small $\varepsilon > 0$ and 
    sufficiently large $N$, the ball 
    $\overline{B}(p, \varepsilon)$
    is mapped by 
    $\overline{f}_{A,b}^N$
    inside the ball 
    $B(p, \varepsilon/2)$.

    In either case, for sufficiently small $\varepsilon > 0$, the ball $\overline{B}(p, \varepsilon)$ is a robust invariant of 
    $\overline{f}_{A,b}^N$ for some $N$ which is contained in the interior of $\widehat{P}$, such that all $x \in \overline{B}(p, \varepsilon)$ remain in the interior of $\widehat{P}$ under iteration of $\overline{f}_{A,b}$, $\overline{f}_{A,b}^k(x_0) \in B(p, \varepsilon)$ for some $k$,
    and $\overline{f}_{A,b}^j(x_0) \in \intr{\widehat{P}}$ for all $j \leq k$.
    It now follows from Propositions \ref{Proposition: pulling back robust invariant} and \ref{Proposition: invariant for n-fold iterate is contained in invariant for f} in the appendix that $\overline{f}_{A,b}(x_0)$ is contained in the interior of a robust invariant of $\overline{f}_{A,b}$ which is contained in the interior of $\widehat{P}$.
    Theorem \ref{Theorem: main theorem} yields that the instance is robust.
\end{proof}

%% file: quadratic-short.tex
\section{The Quadratic Family}\label{Section: Mandelbrot}

    As a final application of Algorithm \ref{Algorithm: main algorithm}, we provide a complete decision method for the Mandelbrot set subject to the hyperbolicity conjecture.
    The \emph{Mandelbrot set} $\mathcal{M}$ is the set of all parameters $c \in \C$ such that the 
    orbit of the origin $0 \in \C$ under the map $f_c(z) = z^2 + c$ is bounded.
    The Mandelbrot set is a compact connected subset of $\C$.
    Despite $\mathcal{M}$ being best known for its numerous computer-generated depictions \cite{PeitgenRichter}, 
    it is unknown whether there exists any rigorous algorithm for computing accurate images of $\mathcal{M}$ to a given resolution.
    See \cite{BravermanYampolsky} for a good discussion of this.
   This has lead Penrose \cite{PenroseEmperor} to go as far as conjecturing that the Mandelbrot set might be an uncomputable subset of $\C$, however without specifying a computational model.
    Shortly after Penrose made his conjecture, Blum and Smale \cite{BlumMandel} showed that the Mandelbrot set is undecidable in the BSS-model of real computation \cite{blum2}. 
    However, Brattka \cite{brattka-emperor} soon raised the objection that BSS-decidability does not adequately capture the intuitive notion of computability that Penrose had in mind:     
    In fact, Penrose explicitly mentions the closed epigraph of the exponential function as an example of an intuitively computable subset of the complex plane.
    However, all BSS-decidable subsets of $\C$ must be semi-algebraic, so that the epigraph of the exponential function is BSS-undecidable.
    Computable analysis offers an alternative definition of computability: a subset $X \subseteq \C$ of the complex plane is called computable if its \emph{distance function}
    $
    d_{X} \colon \C \to \R,
    \;
    d_{X}(c) = \inf \Set{|c - z|}{z \in X}
    $
    is computable. 
    It is quite easy to see that the distance function of a compact set is computable if and only if the set can be plotted at any given resolution -- see for example \cite{BravermanYampolsky} for details.

    Hertling \cite{hertling9} showed that $\mathcal{M}$ is computable in this sense subject to the \emph{hyperbolicity conjecture}, which we describe below.
    He even showed that under this conjecture the \emph{signed distance function}
    $d^s_{\mathcal{M}}(c) = (-1)^{\chi_{\mathcal{M}}(c)}d_{\mathcal{M}}(c)$
    is computable. This is stronger than computability, as there exist sets whose distance function is computable but whose signed distance function is not.
    It follows from Hertling's work that the signed distance function of $\mathcal{M}$ is computable if and only if $\mathcal{M}$ is maximally partially decidable.

    The problem of maximally partially deciding the Mandelbrot set reduces to the Point Escape Problem as follows:
    map $c = a + ib \in \C$ to the Point Escape Problem instance $\left(g_c, \overline{B}(0,3), 0\right)$, where
    $
        g_c(x,y) = \left(x^2 - y^2 + a, 2 xy + b\right)
    $.
    Observe that 
    $g_{c}(x,y) = \left(\operatorname{Re}\left(f_c\left(x + iy\right)\right), \operatorname{Im}\left(f_c\left(x+iy\right)\right)\right)$.

    We will discuss how much of the Mandelbrot set our algorithm for the Point Escape Problem is able to compute under this reduction.
    To state our main result we require one more definition.
    A point $c \in \intr{\mathcal{M}}$ in the interior of $\mathcal{M}$ is called \emph{hyperbolic} if $f_c$ has an 
    \emph{attracting cycle},
    \textit{i.e.} there exist $z_0 \in \C$ and $n \in \N$ with $f^n_c(z_0) = z_0$, $f^j_c(z_0) \neq f^k_c(z_0)$ for $j, k < n$, and
    $\left|\tfrac{d}{dz}(f^n_c) (z_0)\right| < 1$.
    
    The \emph{hyperbolicity conjecture} states that every point in the interior of $\mathcal{M}$ is hyperbolic.        
    We obtain the following result:
            
    \begin{theorem}
        The decision problem for the Mandelbrot set reduces to the General Escape Problem via the reduction above.
        This reduction maps a parameter $c \in \C$ to a robust instance of the General Escape Problem
        if and only if $c$ belongs to $\C \setminus \mathcal{M}$ or to a hyperbolic component of the interior of $\mathcal{M}$. 
        In particular, the reduction maps robust instances to robust instances if and only if the hyperbolicity conjecture holds true.
    \end{theorem}
    \begin{proof}
        Consider the map $c \mapsto \left(g_c, \overline{B}(0,3), 0\right)$ defined above.
        It is easy to see that this is really a reduction.

        It is obvious that -- under this reduction -- Algorithm \ref{Algorithm: main algorithm} halts on all parameters outside $\mathcal{M}$
        and that Algorithm \ref{Algorithm: main algorithm} does not halt on any boundary point.
        
        Now, let $c \in \intr{\mathcal{M}}$.
        Assume first that $c$ is hyperbolic. Then $c$ has an attracting cycle $\{z_0, z_1,\dots, z_{n - 1}\}$ as described above.
        It follows from a celebrated theorem by Fatou \cite[Theorem 1]{Branner} that the origin is attracted by this attracting cycle,
        \textit{i.e.} for all $\varepsilon$ there exists $N$ such that 
        $\left|f^N_c(0) - z_0\right| < \varepsilon$.

        We show that the attracting cycle induces a robust invariant.
        Since $|\tfrac{d}{dz}(f^n_c)(z_0)|  < 1$, there exists $\delta > 0$ such that $|\tfrac{d}{dz}(f^n_c)(z_0)|  < 1 - \delta$.
        Further, observe that by the chain rule we have $\tfrac{d}{dz}(f^n_c) (z_0) = \tfrac{d}{dz}(f^n_c) (f^k(z_0))$ for all $k$.
        By continuity, there exists $\varepsilon > 0$ such that 
        $|\tfrac{d}{dz}(f^n_c) (z)| < 1 - \delta$ for all $z$ with $|z - f^k(z_0)| < \varepsilon$ for some $k \in \N$.
        Consider the set 
        $
            K = \bigcup_{j = 0}^{n - 1} \overline{B}(f^{k}_c(z_0), \varepsilon).
        $
        It follows from the uniform bound on the first derivative that on $K$, the function $f^{n}_c$ is Lipschitz-continuous with Lipschitz constant $1 - \delta$.
        Let $z \in K$ with $|z - f^{k}_c(z_0)| < \varepsilon$.
        Then we obtain:
        \[
            \left| f^n_c\left(f_c(z)\right) - f^{k}_c(z_0)\right| 
            =
            \left| f^n_c\left(f_c(z)\right) - f^{n}_c\left(f^{k}_c(z_0)\right)\right| 
            \leq 
            \left|1 - \delta\right| 
            \left|f^n_c(z) - f^{k}_c(z_0)\right|
            \leq
            \left|1 - \delta\right| \varepsilon
            < \varepsilon.
        \]
        So $K$ is a robust invariant for $f^n_c$.
        Since $0$ is attracted by the attracting cycle, we have $f^N(0) \in \intr{K}$ for sufficiently large $N$.
        It now follows from Propositions \ref{Proposition: pulling back robust invariant} and \ref{Proposition: invariant for n-fold iterate is contained in invariant for f} in the appendix that there exists a robust invariant for $f_c$ whose interior contains $0$.
        This invariant must be contained in the closed disk of radius $2$, for any orbit of $f_c$ that contains an element of absolute value greater than $2$ must be unbounded. In particular, it must be contained in the open disk of radius $3$. 
        Hence, Algorithm \ref{Algorithm: main algorithm} halts on input $(g_c, \overline{B}(0,3), 0)$.

Suppose on the other hand that $c$ belongs to a non-hyperbolic component of the interior of $\mathcal{M}$.
Let $K_c = \Set{z \in \C}{ f^n_c(z) \text{ is bounded}}$ denote the \emph{filled-in Julia set} of $c$.
It can be shown (see \cite{Orsay}) that since $c$ belongs to a non-hyperbolic component of the interior of $\mathcal{M}$,
we must have $c \in K_c = \partial K_c$. 
This implies that for all $\varepsilon > 0$ there exists 
$z_\varepsilon$ 
with 
$|z_{\varepsilon} - c| < \varepsilon$
and 
$z_{\varepsilon} \notin K_c$.
Since the orbit $f^n_c(z_\varepsilon)$ is unbounded and disjoint from $K_c$ with $0 \in K_c$, there exists a number $\delta > 0$ with 
$|f^n_c(z_\varepsilon)| > \delta$.
Thus, if we define 
\[
    f_{\varepsilon}(z) = 
        \begin{cases}
            z_\varepsilon\left(1 - \frac{|z|}{\delta}\right) + f_c(z) \frac{|z|}{\delta} &\text{if }|z| < \delta,\\
            f_c(z)                                         &\text{otherwise.}
        \end{cases}
\]
then the orbit of $0$ under $f_{\varepsilon}$ is unbounded, and $\norm{f_{\varepsilon} - f_c}_\infty < \varepsilon$.
This shows that Algorithm \ref{Algorithm: main algorithm} cannot halt on input $\left(g_c, \overline{B}(0,3), 0\right)$.
    \end{proof}

%% file: Escape-Continuous.bbl
\begin{thebibliography}{10}

\bibitem{MarkovChains}
Manindra Agrawal, S.~Akshay, Blaise Genest, and P.S. Thiagarajan.
\newblock {Approximate Verification of the Symbolic Dynamics of Markov Chains}.
\newblock In {\em 2012 27th Annual IEEE Symposium on Logic in Computer
  Science}, pages 55--64, 2012.
\newblock \href {https://doi.org/10.1109/LICS.2012.17}
  {\path{doi:10.1109/LICS.2012.17}}.

\bibitem{Robust}
S.~Akshay, Hugo Bazille, Blaise Genest, and Mihir Vahanwala.
\newblock {On Robustness for the Skolem and Positivity Problems}.
\newblock In Petra Berenbrink and Benjamin Monmege, editors, {\em 39th
  International Symposium on Theoretical Aspects of Computer Science, {STACS}
  2022, March 15-18, 2022, Marseille, France (Virtual Conference)}, volume 219
  of {\em LIPIcs}, pages 5:1--5:20. Schloss Dagstuhl - Leibniz-Zentrum
  f{\"{u}}r Informatik, 2022.
\newblock \href {https://doi.org/10.4230/LIPIcs.STACS.2022.5}
  {\path{doi:10.4230/LIPIcs.STACS.2022.5}}.

\bibitem{Rounding}
Christel Baier, Florian Funke, Simon Jantsch, Toghrul Karimov, Engel
  Lefaucheux, Jo{\"{e}}l Ouaknine, Amaury Pouly, David Purser, and Markus~A.
  Whiteland.
\newblock {Reachability in Dynamical Systems with Rounding}.
\newblock In Nitin Saxena and Sunil Simon, editors, {\em 40th {IARCS} Annual
  Conference on Foundations of Software Technology and Theoretical Computer
  Science, {FSTTCS} 2020, December 14-18, 2020, {BITS} Pilani, {K} {K} Birla
  Goa Campus, Goa, India (Virtual Conference)}, volume 182 of {\em LIPIcs},
  pages 36:1--36:17. Schloss Dagstuhl - Leibniz-Zentrum f{\"{u}}r Informatik,
  2020.
\newblock \href {https://doi.org/10.4230/LIPIcs.FSTTCS.2020.36}
  {\path{doi:10.4230/LIPIcs.FSTTCS.2020.36}}.

\bibitem{Economics}
William~J. Baumol.
\newblock {\em {Economic Dynamics}}.
\newblock Prentice Hall, 1970.

\bibitem{BellGerhold07}
Jason~P. Bell and Stefan Gerhold.
\newblock {On the positivity set of a linear recurrence}.
\newblock {\em Israel Journal of Mathematics}, 157:333 -- 345, 2007.
\newblock \href {https://doi.org/10.1007/s11856-006-0015-1}
  {\path{doi:10.1007/s11856-006-0015-1}}.

\bibitem{Quantum}
Vincent~D. Blondel, Emmanuel Jeandel, Pascal Koiran, and Natacha Portier.
\newblock {Decidable and Undecidable Problems about Quantum Automata}.
\newblock {\em SIAM Journal on Computing}, 34(6):1464--1473, 2005.
\newblock \href {https://doi.org/10.1137/S0097539703425861}
  {\path{doi:10.1137/S0097539703425861}}.

\bibitem{blum2}
Lenore Blum, Felipe Cucker, Michael Shub, and Steve Smale.
\newblock {\em {Complexity and Real Computation}}.
\newblock Springer, 1998.

\bibitem{BlumMandel}
Lenore Blum and Steve Smale.
\newblock {\em {The G{\"o}del Incompleteness Theorem and Decidability over a
  Ring}}, pages 321--339.
\newblock Springer US, New York, NY, 1993.
\newblock \href {https://doi.org/10.1007/978-1-4612-2740-3_32}
  {\path{doi:10.1007/978-1-4612-2740-3_32}}.

\bibitem{Branner}
Bodil Branner.
\newblock {\em {Chaos and Fractals: The Mathematics Behind the Computer
  Graphics}}, chapter The Mandelbrot Set, pages 75 -- 106.
\newblock American Mathematical Society, 1989.

\bibitem{brattka-emperor}
Vasco Brattka.
\newblock {The emperor's new recursiveness: {T}he epigraph of the exponential
  function in two models of computability}.
\newblock In Masami Ito and Teruo Imaoka, editors, {\em Words, Languages \&
  Combinatorics III}, pages 63--72, 2003.

\bibitem{BrattkaHertling21}
Vasco Brattka and Peter Hertling, editors.
\newblock {\em {Handbook of Computability and Complexity in Analysis}}, Theory
  and Applications of Computability, Cham, 2021. Springer.
\newblock \href {https://doi.org/10.1007/978-3-030-59234-9}
  {\path{doi:10.1007/978-3-030-59234-9}}.

\bibitem{BrattkaPresser}
Vasco Brattka and Gero Presser.
\newblock {Computability on subsets of metric spaces}.
\newblock {\em Theoretical Computer Science}, 305(1 -- 3):43 -- 76, 2003.

\bibitem{BravermanYampolsky}
Mark Braverman and Michael Yampolsky.
\newblock {\em {Computability of Julia Sets}}.
\newblock Springer, 2009.

\bibitem{Pseudo}
Julian D'Costa, Toghrul Karimov, Rupak Majumdar, Jo{\"{e}}l Ouaknine, Mahmoud
  Salamati, Sadegh Soudjani, and James Worrell.
\newblock {The Pseudo-Skolem Problem is Decidable}.
\newblock In Filippo Bonchi and Simon~J. Puglisi, editors, {\em 46th
  International Symposium on Mathematical Foundations of Computer Science,
  {MFCS} 2021, August 23-27, 2021, Tallinn, Estonia}, volume 202 of {\em
  LIPIcs}, pages 34:1--34:21. Schloss Dagstuhl - Leibniz-Zentrum f{\"{u}}r
  Informatik, 2021.
\newblock \href {https://doi.org/10.4230/LIPIcs.MFCS.2021.34}
  {\path{doi:10.4230/LIPIcs.MFCS.2021.34}}.

\bibitem{Orsay}
Adrien Douady and John~H. Hubbard.
\newblock {\em {{\'E}tude dynamique des polyn{\^o}mes complexes}}.
\newblock Publications math\'ematiques d'Orsay, 1984.

\bibitem{GerholdKauers05}
Stefan Gerhold and Manuel Kauers.
\newblock A procedure for proving special function inequalities involving a
  discrete parameter.
\newblock In Manuel Kauers, editor, {\em Symbolic and Algebraic Computation,
  International Symposium {ISSAC} 2005, Beijing, China, July 24-27, 2005,
  Proceedings}, pages 156--162. {ACM}, 2005.

\bibitem{hertling9}
Peter Hertling.
\newblock {Is the Mandelbrot set computable?}
\newblock {\em Mathematical Logic Quarterly}, 51(1):5--18, 2005.
\newblock URL: \url{http://dx.doi.org/10.1002/malq.200310124}, \href
  {https://doi.org/10.1002/malq.200310124} {\path{doi:10.1002/malq.200310124}}.

\bibitem{HoyrupRojasMarkov}
Mathieu Hoyrup and Crist\'{o}bal Rojas.
\newblock {On the Information Carried by Programs About the Objects they
  Compute}.
\newblock {\em Theor. Comp. Sys.}, 61(4):1214--1236, November 2017.
\newblock \href {https://doi.org/10.1007/s00224-016-9726-9}
  {\path{doi:10.1007/s00224-016-9726-9}}.

\bibitem{KauersPillwein10}
Manuel Kauers and Veronika Pillwein.
\newblock {When can we detect that a {P}-finite sequence is positive?}
\newblock In Wolfram Koepf, editor, {\em Symbolic and Algebraic Computation,
  International Symposium, {ISSAC} 2010, Munich, Germany, July 25-28, 2010,
  Proceedings}, pages 195--201. {ACM}, 2010.

\bibitem{kenison2024threshold}
George Kenison.
\newblock {The Threshold Problem for Hypergeometric Sequences with Quadratic
  Parameters}.
\newblock In Karl Bringmann, Martin Grohe, Gabriele Puppis, and Ola Svensson,
  editors, {\em 51st International Colloquium on Automata, Languages, and
  Programming (ICALP 2024)}, volume 297 of {\em Leibniz International
  Proceedings in Informatics (LIPIcs)}, pages 145:1--145:20, Dagstuhl, Germany,
  2024. Schloss Dagstuhl -- Leibniz-Zentrum f{\"u}r Informatik.
\newblock URL:
  \url{https://drops.dagstuhl.de/entities/document/10.4230/LIPIcs.ICALP.2024.145},
  \href {https://doi.org/10.4230/LIPIcs.ICALP.2024.145}
  {\path{doi:10.4230/LIPIcs.ICALP.2024.145}}.

\bibitem{HolonomicPeriods20}
George Kenison, Oleksiy Klurman, Engel Lefaucheux, Florian Luca, Pieter Moree,
  Jo\"{e}l Ouaknine, Markus~A. Whiteland, and James Worrell.
\newblock {On Positivity and Minimality for Second-Order Holonomic Sequences}.
\newblock {\em CoRR}, abs/2007.12282, 2020.
\newblock URL: \url{https://arxiv.org/abs/2007.12282}.

\bibitem{Ko}
Ker-I Ko.
\newblock {\em {Complexity theory of real functions}}.
\newblock Progress in Theoretical Computer Science. Birkh\"auser Boston, Inc.,
  Boston, MA, 1991.
\newblock \href {https://doi.org/10.1007/978-1-4684-6802-1}
  {\path{doi:10.1007/978-1-4684-6802-1}}.

\bibitem{Kushner}
Boris~A. Kushner.
\newblock {Markov's constructive analysis; a participant's view}.
\newblock {\em Theoretical Computer Science}, 219(1):267--285, 1999.
\newblock \href {https://doi.org/10.1016/S0304-3975(98)00291-6}
  {\path{doi:10.1016/S0304-3975(98)00291-6}}.

\bibitem{MelczerMezzarobba}
Stephen Melczer and Marc Mezzarobba.
\newblock {Sequence positivity through numeric analytic continuation:
  uniqueness of the {C}anham model for biomembranes}.
\newblock {\em Comb. Theory}, 2(2):Paper No. 4, 20, 2022.

\bibitem{DPLRSR21}
Eike Neumann.
\newblock {Decision problems for linear recurrences involving arbitrary real
  numbers}.
\newblock {\em {Logical Methods in Computer Science}}, {Volume 17, Issue 3},
  August 2021.
\newblock URL: \url{https://lmcs.episciences.org/8346}, \href
  {https://doi.org/10.46298/lmcs-17(3:16)2021}
  {\path{doi:10.46298/lmcs-17(3:16)2021}}.

\bibitem{DBLP:conf/icalp/NeumannO021}
Eike Neumann, Jo{\"{e}}l Ouaknine, and James Worrell.
\newblock {Decision Problems for Second-Order Holonomic Recurrences}.
\newblock In Nikhil Bansal, Emanuela Merelli, and James Worrell, editors, {\em
  48th International Colloquium on Automata, Languages, and Programming,
  {ICALP} 2021, July 12-16, 2021, Glasgow, Scotland (Virtual Conference)},
  volume 198 of {\em LIPIcs}, pages 99:1--99:20. Schloss Dagstuhl -
  Leibniz-Zentrum f{\"{u}}r Informatik, 2021.
\newblock \href {https://doi.org/10.4230/LIPIcs.ICALP.2021.99}
  {\path{doi:10.4230/LIPIcs.ICALP.2021.99}}.

\bibitem{OWPositivitySimple14}
Jo{\"{e}}l Ouaknine and James Worrell.
\newblock {On the Positivity Problem for Simple Linear Recurrence Sequences}.
\newblock In Javier Esparza, Pierre Fraigniaud, Thore Husfeldt, and Elias
  Koutsoupias, editors, {\em Automata, Languages, and Programming - 41st
  International Colloquium, {ICALP} 2014, Copenhagen, Denmark, July 8-11, 2014,
  Proceedings, Part {II}}, volume 8573 of {\em Lecture Notes in Computer
  Science}, pages 318--329. Springer, 2014.
\newblock \href {https://doi.org/10.1007/978-3-662-43951-7\_27}
  {\path{doi:10.1007/978-3-662-43951-7\_27}}.

\bibitem{OWPositivityLowOrder14}
Jo{\"{e}}l Ouaknine and James Worrell.
\newblock {Positivity Problems for Low-Order Linear Recurrence Sequences}.
\newblock In Chandra Chekuri, editor, {\em Proceedings of the Twenty-Fifth
  Annual {ACM-SIAM} Symposium on Discrete Algorithms, {SODA} 2014, Portland,
  Oregon, USA, January 5-7, 2014}, pages 366--379. {SIAM}, 2014.
\newblock \href {https://doi.org/10.1137/1.9781611973402.27}
  {\path{doi:10.1137/1.9781611973402.27}}.

\bibitem{OWUltimateSimple14}
Jo{\"{e}}l Ouaknine and James Worrell.
\newblock {Ultimate Positivity is Decidable for Simple Linear Recurrence
  Sequences}.
\newblock In Javier Esparza, Pierre Fraigniaud, Thore Husfeldt, and Elias
  Koutsoupias, editors, {\em Automata, Languages, and Programming - 41st
  International Colloquium, {ICALP} 2014, Copenhagen, Denmark, July 8-11, 2014,
  Proceedings, Part {II}}, volume 8573 of {\em Lecture Notes in Computer
  Science}, pages 330--341. Springer, 2014.
\newblock \href {https://doi.org/10.1007/978-3-662-43951-7\_28}
  {\path{doi:10.1007/978-3-662-43951-7\_28}}.

\bibitem{OWSurvey15}
Jo{\"{e}}l Ouaknine and James Worrell.
\newblock On linear recurrence sequences and loop termination.
\newblock {\em {ACM} {SIGLOG} News}, 2(2):4--13, 2015.
\newblock URL: \url{https://dl.acm.org/citation.cfm?id=2766191}.

\bibitem{PaulyRepresented}
Arno Pauly.
\newblock On the topological aspects of the theory of represented spaces.
\newblock {\em Computability}, 5(2):159--180, 2016.
\newblock URL: \url{http://arxiv.org/abs/1204.3763}, \href
  {https://doi.org/10.3233/COM-150049} {\path{doi:10.3233/COM-150049}}.

\bibitem{PeitgenRichter}
Heinz-Otto Peitgen and Peter Richter.
\newblock {\em {The Beauty of Fractals}}.
\newblock Springer-Verlag, Heidelberg, 1986.

\bibitem{PenroseEmperor}
Roger Penrose.
\newblock {\em {The Emperor's New Mind}}.
\newblock Oxford University Press, 1989.

\bibitem{PillweinJacobi}
Veronika Pillwein.
\newblock {Positivity of certain sums over Jacobi kernel polynomials}.
\newblock {\em Advances in Applied Mathematics}, 41(3), 2008.

\bibitem{Pillwein13}
Veronika Pillwein.
\newblock {Termination Conditions for Positivity Proving Procedures}.
\newblock In {\em Proceedings of the 38th International Symposium on Symbolic
  and Algebraic Computation}, ISSAC '13, page 315–322, New York, NY, USA,
  2013. Association for Computing Machinery.

\bibitem{PillweinSchussler15}
Veronika Pillwein and Miriam Schussler.
\newblock {An efficient procedure deciding positivity for a class of holonomic
  sequences}.
\newblock {\em ACM Communications in Computer Algebra}, 49(3):90--93, 2015.
\newblock Extended abstract of the poster presentation at ISSAC 2015.

\bibitem{Ratschan}
Stefan Ratschan.
\newblock {Safety verification of non-linear hybrid systems is
  quasi-decidable}.
\newblock {\em Formal Methods in System Design}, 44(1):71--90, 2014.
\newblock \href {https://doi.org/10.1007/s10703-013-0196-2}
  {\path{doi:10.1007/s10703-013-0196-2}}.

\bibitem{SchroederPhD}
Matthias Schr\"oder.
\newblock {\em {Admissible {R}epresentations for {C}ontinuous {C}omputations}}.
\newblock PhD thesis, Fern{U}niversit\"at {H}agen, 2002.

\bibitem{Szegö1933}
G.~Szegö.
\newblock {Über gewisse Potenzreihen mit lauter positiven Koeffizienten.}
\newblock {\em Mathematische Zeitschrift}, 37:674--688, 1933.
\newblock URL: \url{http://eudml.org/doc/168480}.

\bibitem{Mignotte1984}
R.~Tijdeman, M.~Mignotte, and T.N. Shorey.
\newblock {The distance between terms of an algebraic recurrence sequence}.
\newblock {\em Journal f\"ur die reine und angewandte Mathematik}, 349:63--76,
  1984.
\newblock URL: \url{http://eudml.org/doc/152622}.

\bibitem{Tiwari}
Ashish Tiwari.
\newblock {Termination of Linear Programs}.
\newblock In {\em CAV'04}, volume 3114 of {\em LNCS}. Springer, 2004.

\bibitem{Vereshchagin}
N.K. Vereshchagin.
\newblock Occurrence of zero in a linear recursive sequence.
\newblock {\em Mat. Zametki}, 38(2):177 -- 189, 1985.

\bibitem{WeihrauchBook}
Klaus Weihrauch.
\newblock {\em {Computable Analysis}}.
\newblock Springer, 2000.

\bibitem{YuChen}
Thomas Yu and Jingmin Chen.
\newblock {Uniqueness of Clifford torus with prescribed isoperimetric ratio}.
\newblock {\em Proceedings of the American Mathematical Society}, 150, 2022.
\newblock \href {https://doi.org/10.1090/proc/15750}
  {\path{doi:10.1090/proc/15750}}.

\end{thebibliography}
